\documentclass[preprint,12pt]{article}
\usepackage{lineno,hyperref}
\modulolinenumbers[5]

\usepackage[utf8]{inputenc}
\usepackage[spanish,english]{babel}
\usepackage{graphicx}
\usepackage{amssymb}
\usepackage{amsmath}
\usepackage{hyperref}
\usepackage{color}
\usepackage{latexsym,dsfont,mathtools}
\usepackage{theorem,titletoc,titlesec}
\usepackage[normalem]{ulem}
\usepackage{enumitem}
\usepackage{tikz}
\usepackage{tabularx}
\usetikzlibrary{matrix}
\usepackage{adjustbox}
\usepackage{soul}
\newtheorem{theorem}{Theorem}[section]
\newtheorem{corollary}{Corollary}[section]
\newtheorem{proposition}{Proposition}[section]
\newtheorem{lemma}{Lemma}[section]
\newtheorem{example}{Example}[section]
\newtheorem{remark}{Remark}[section]

\newenvironment{proof}[1][Proof.]{\vspace{0.5em}\textbf{#1} }{\
\hfill\rule{0.5em}{0.5em}}

\usepackage{multirow}
\usepackage{tabularx}
\usepackage{xcolor}

\newcommand{\Z}{\mathbb{Z}}

\newcommand{\zero}{{\mathbf{0}}}
\newcommand{\one}{{\mathbf{1}}}

\newcommand{\C}{{\cal C}}

\newcommand{\wt}{{\rm wt}}

\newcommand{\two}{{\mathbf{2}}}
\newcommand{\four}{{\mathbf{4}}}

\newcommand{\uu}{\mathbf{u}}
\newcommand{\vv}{\mathbf{v}}
\newcommand{\ww}{\mathbf{w}}
\newcommand{\zz}{\mathbf{z}}
\newcommand{\rank}{\operatorname{rank}}
\newcommand{\kernel}{\operatorname{ker}}

\newcommand{\cA}{{\cal A}}
\newcommand{\cH}{\cal{H}}
\newcommand{\cG}{{\cal G}}

\begin{document}


\title{$\mathbb{Z}_2\mathbb{Z}_4\mathbb{Z}_8$-Additive Hadamard Codes
\thanks{This work has been partially supported by the Spanish MINECO under Grant PID2019-104664GB-I00
(AEI / 10.13039/501100011033) and by Catalan AGAUR scholarship 2020 FI SDUR 00475.}
}



\author{Dipak K. Bhunia, Cristina Fern\'andez-C\'ordoba, Merc\`e Villanueva}

\maketitle



\begin{abstract}
The $\mathbb{Z}_2\mathbb{Z}_4\mathbb{Z}_8$-additive codes are subgroups of $\mathbb{Z}_2^{\alpha_1} \times \mathbb{Z}_4^{\alpha_2} \times \mathbb{Z}_8^{\alpha_3}$, and can be seen as linear codes over $\mathbb{Z}_2$ when $\alpha_2=\alpha_3=0$,  $\mathbb{Z}_4$-additive or $\mathbb{Z}_8$-additive codes when $\alpha_1=\alpha_3=0$ or $\alpha_1=\alpha_2=0$, respectively, or $\mathbb{Z}_2\mathbb{Z}_4$-additive codes when $\alpha_3=0$. A $\mathbb{Z}_2\mathbb{Z}_4\mathbb{Z}_8$-linear  Hadamard code is a Hadamard code which is the Gray map image of a $\mathbb{Z}_2\mathbb{Z}_4\mathbb{Z}_8$-additive code.  
In this paper, we generalize some known results for $\mathbb{Z}_2\mathbb{Z}_4$-linear Hadamard codes to $\mathbb{Z}_2\mathbb{Z}_4\mathbb{Z}_8$-linear Hadamard codes with $\alpha_1 \neq 0$, $\alpha_2 \neq 0$, and $\alpha_3 \neq 0$. First, we give a recursive construction of $\mathbb{Z}_2\mathbb{Z}_4\mathbb{Z}_8$-additive Hadamard codes of type $(\alpha_1,\alpha_2, \alpha_3;t_1,t_2, t_3)$ with $t_1\geq 1$, $t_2 \geq 0$, and $t_3\geq 1$.
Then, we show that in general the $\mathbb{Z}_4$-linear, $\mathbb{Z}_8$-linear and $\mathbb{Z}_2\mathbb{Z}_4$-linear Hadamard codes are not included in the family of $\mathbb{Z}_2\mathbb{Z}_4\mathbb{Z}_8$-linear Hadamard codes with $\alpha_1 \neq 0$, $\alpha_2 \neq 0$, and $\alpha_3 \neq 0$. Actually, we point out that  none of these nonlinear $\mathbb{Z}_2\mathbb{Z}_4\mathbb{Z}_8$-linear Hadamard codes of length $2^{11}$ is equivalent to a $\mathbb{Z}_2\mathbb{Z}_4\mathbb{Z}_8$-linear Hadamard code of any other type, a $\mathbb{Z}_2\mathbb{Z}_4$-linear Hadamard code, or a $\mathbb{Z}_{2^s}$-linear Hadamard code, with $s\geq 2$, of the same length $2^{11}$.
\end{abstract}


\section{Introduction}
Let $\Z_{2^s}$ be the ring of integers modulo $2^s$ with $s\geq1$. The set of
$n$-tuples over $\Z_{2^s}$ is denoted by $\Z_{2^s}^n$. In this paper,
the elements of $\Z^n_{2^s}$ will also be called vectors. 
A code over $\Z_2$ of length $n$ is a nonempty subset of $\Z_2^n$,
and it is linear if it is a subspace of $\Z_{2}^n$. Similarly, a nonempty
subset of $\Z_{2^s}^n$ is a $\Z_{2^s}$-additive code if it is a subgroup of $\Z_{2^s}^n$. A $\Z_2\Z_4\Z_8$-additive code is a subgroup of $\Z_2^{\alpha_1}  \times \Z_4^{\alpha_2} \times \Z_8^{\alpha_3}$. Note that a $\Z_2\Z_4\Z_8$-additive code is a linear code over $\Z_2$ when $\alpha_2=\alpha_3=0$,  a $\Z_4$-additive or $\Z_8$-additive code when $\alpha_1=\alpha_3=0$ or $\alpha_1=\alpha_2=0$, respectively, and a $\Z_2\Z_4$-additive code when $\alpha_3=0$. 
The order of a vector $u\in \Z_{2^s}$, denoted by $o(u)$, is the smallest positive integer $m$ such that $m u =(0,\dots,0)$. Also, the order of a vector $\mathbf u\in \Z_2^{\alpha_1}\times\Z_4^{\alpha_2} \times\Z_8^{\alpha_3}$, denoted by $o(\mathbf u)$, is the smallest positive integer $m$ such that $m \mathbf u =(0,\dots,0\mid 0,\dots,0\mid 0,\dots,0)$.

The Hamming weight of a vector $u\in\Z_{2}^n$, denoted by $\wt_H(u)$, is
the number of nonzero coordinates of $u$. The Hamming distance of two
vectors $u,v\in\Z_{2}^n$, denoted by $d_H(u,v)$, is the number of
coordinates in which they differ.  Note that $d_H(u,v)=\wt_H(u-v)$. The minimum distance of a code $C$ over $\Z_2$ is $d(C)=\min \{ d_H(u,v) : u,v \in C, u \not = v  \}$.

In \cite{Sole}, a Gray map  from $\Z_4$ to $\Z_2^2$ is defined as
$\phi(0)=(0,0)$, $\phi(1)=(0,1)$, $\phi(2)=(1,1)$ and $\phi(3)=(1,0)$. There exist different generalizations of this Gray map, which go from $\Z_{2^s}$ to
$\Z_2^{2^{s-1}}$ \cite{Carlet,Codes2k,dougherty,Nechaev,Krotov:2007}.
The one given in \cite{Nechaev} can be defined in terms of the elements of a Hadamard code \cite{Krotov:2007}, and Carlet's Gray map \cite{Carlet} is a particular case of the one given in \cite{Krotov:2007} 
satisfying $\sum \lambda_i \phi(2^i) =\phi(\sum \lambda_i 2^i)$ \cite{KernelZ2s}. 
In this paper, we focus on Carlet's Gray map \cite{Carlet}, from $\Z_{2^s}$ to $\Z_2^{2^{s-1}}$, which is also a particular case of the one given in \cite{ShiKrotov2019}. Specifically, \begin{gather}\label{eq:GrayMapCarlet}
\phi_s(u)=(u_{s-1},u_{s-1},\dots,u_{s-1})+(u_0,\dots,u_{s-2})Y_{s-1},
\end{gather}
where $u\in\Z_{2^s}$; $[u_0,u_1,\dots,u_{s-1}]_2$ is the binary expansion of $u$, that is, $u=\sum_{i=0}^{s-1}u_i2^i$ with $u_i\in \{0,1\}$; and $Y$ is the  matrix of size $(s-1)\times 2^{s-1}$ whose columns are all the vectors in $\Z_2^{s-1}$. Without loss of generality, we assume that the columns of $Y_{s-1}$ are ordered in ascending order, by considering the elements of $\mathbb{Z}_2^{s-1}$ as the binary expansions of the elements of $\mathbb{Z}_{2^{s-1}}$. Note that $\phi_1$ is the identity map, and $(u_{s-1},\dots,u_{s-1})$ and $(u_0,\dots,u_{s-2})Y_{s-1}$ are binary vectors of length $2^{s-1}$, and that the rows of $Y_{s-1}$ form a basis of a first order Reed-Muller code
after adding the all-one row. We define $\Phi_s:\Z_{2^s}^n \rightarrow \Z_2^{n2^{s-1}}$ as the component-wise extended map of $\phi_s$.
We can also define a Gray map $\Phi$ from $\Z_2^{\alpha_1}  \times \Z_4^{\alpha_2} \times \Z_8^{\alpha_3}$ to $\Z_2^n$, where $n=\alpha_1+2\alpha_2+4\alpha_3$, as follows:
$$
\Phi(u_1\mid u_2 \mid u_3)=(u_1, \Phi_2(u_2),\Phi_3(u_3)),
$$
for any $u_i \in \Z_{2^i}^{\alpha_i}$, where $1\leq i\leq 3$.

Let $\C \subseteq \Z_{2^s}^n$ be a $\Z_{2^s}$-additive code of length $n$. We say that its Gray map image
$C=\Phi_s(\C)$ is a $\Z_{2^s}$-linear code of length $n2^{s-1}$. 
Since $\C$ is a subgroup of
$\Z_{2^s}^n$, it is isomorphic to an abelian structure
$\Z_{2^s}^{t_1}\times\Z_{2^{s-1}}^{t_2}\times
\dots\times\Z_p^{t_s}$, and we say that $\C$, or equivalently
$C=\Phi_s(\C)$, is of type $(n;t_1,\dots,t_{s})$.
Note that $|\C|=p^{st_1}p^{(s-1)t_2}\cdots p^{t_s}$. Similarly, if $\C \subseteq \Z_2^{\alpha_1}  \times \Z_4^{\alpha_2} \times \Z_8^{\alpha_3}$
is a $\Z_2\Z_4\Z_8$-additive code,  we say that its Gray map image $C=\Phi(\C)$ is a $\Z_2\Z_4\Z_8$-linear code of length $\alpha_1+2\alpha_2+4\alpha_3$.
Since $\C$ can be seen as a subgroup of
$\Z_8^{\alpha_1+\alpha_2+\alpha_3}$, it is isomorphic to an abelian structure
$\Z_8^{t_1}\times\Z_4^{t_2} \times\Z_2^{t_3}$, and we say that $\C$, or equivalently
$C=\Phi(\C)$, is of type $(\alpha_1, \alpha_2,\alpha_3;t_1,t_2,t_3)$.
Note that $|\C|=8^{t_1}4^{t_2} 2^{t_3}$. Unlike linear codes over finite fields,
linear codes over rings do not have a basis, but there
exists a generator matrix for these codes having minimum number of rows, that is, $t_1+\dots+t_s$ rows.
If $\alpha_1=\alpha_3=0$ (respectively, $\alpha_1=\alpha_2=0$), then they coincide with $\Z_4$-additive codes (respectively, $\Z_8$-additive codes). If $\alpha_3=0$, then they are also known as $\Z_2\Z_4$-additive codes, and their Gray map images as $\Z_2\Z_4$-linear codes. In the last case, we also say that the code, or equivalently  the Gray map image of the code, is of type $(\alpha_1, \alpha_2;t_1,t_2)$. Note that there are no $\Z_2\Z_4\Z_8$-linear Hadamard codes neither with only $\alpha_1=0$  nor with only $\alpha_2=0$ \cite{Krotov:2007, ShiKrotov2019}.

Two structural properties of codes over $\Z_2$ are the rank and
dimension of the kernel. The rank of a code $C$ over $\Z_2$ is simply the
dimension of the linear span, $\langle C \rangle$,  of $C$.
The kernel of a code $C$ over $\Z_2$ is defined as
$\mathrm{K}(C)=\{\textbf{x}\in \Z_2^n : \textbf{x}+C=C \}$ \cite{BGH83,pKernel}. If the all-zero vector belongs to $C$,
then $\mathrm{K}(C)$ is a linear subcode of $C$.
Note also that if $C$ is linear, then $K(C)=C=\langle C \rangle$.
We denote the rank of $C$ as $\rank(C)$ and the dimension of the kernel as $\kernel(C)$.
These parameters can be used to distinguish between nonequivalent codes, since equivalent ones have the same rank and dimension of the kernel.

A binary code of length $n$, $2n$ codewords and minimum distance $n/2$ is called a Hadamard code. Hadamard codes can be constructed from Hadamard matrices \cite{Key,WMcwill}.
Note that linear Hadamard codes are in fact first order Reed-Muller codes, or
equivalently, the dual of extended Hamming codes \cite[Ch.13 \S 3]{WMcwill}. The $\Z_{2^s}$-additive codes such that after the Gray map $\Phi_s$ give
Hadamard codes are called $\Z_{2^s}$-additive Hadamard codes and the
corresponding images are called $\Z_{2^s}$-linear Hadamard codes. 
Similarly, the $\Z_2\Z_4\Z_8$-additive codes such that after the Gray map $\Phi$ give
Hadamard codes are called $\Z_2\Z_4\Z_8$-additive Hadamard codes and the
corresponding images are called $\Z_2\Z_4\Z_8$-linear Hadamard codes. 

It is known that $\Z_4$-linear Hadamard codes (that is, $\Z_2\Z_4$-linear Hadamard code with $\alpha_1 =0$) and $\Z_2\Z_4$-linear Hadamard codes with $\alpha_1\not = 0$ can be classified by using either the rank or the dimension of the kernel \cite{Kro:2001:Z4_Had_Perf,PRV06}. Moreover, in \cite{KV2015}, it is shown that each $\Z_4$-linear Hadamard code is  equivalent to a  $\Z_2\Z_4$-linear Hadamard code with $\alpha_1 \not =0$.
Later, in \cite{KernelZ2s,HadamardZps,EquivZ2s,ZpsEquivalance}, an iterative construction for $\Z_{p^s}$-linear Hadamard codes is described, the linearity is established, and a partial classification by using the dimension of the kernel is obtained,
giving the exact amount of nonequivalent such codes for some parameters. 
In \cite{fernandez2019mathbb}, a complete classification of $\Z_8$-linear Hadamard codes by using the rank and dimension of the kernel is provided, giving the exact amount of nonequivalent such codes.  For any $t\geq 2$, the full classification of $\Z_p\Z_{p^2}$-linear Hadamard codes of length $p^t$, with $\alpha_1\neq 0$, $\alpha_2\neq 0$, and $p\geq 3$ prime, is given in \cite{ZpZp2Construction,WCC2022,ZpZp2Classification}, by using just the dimension of the kernel. 


  This paper is focused on $\Z_2\Z_4\Z_8$-linear Hadamard codes with $\alpha_1\not =0$, $\alpha_2\not =0$, and $\alpha_3\not =0$, generalizing some results given for $\Z_2\Z_4$-linear Hadamard codes with $\alpha_1\not =0$  in \cite{PRV06,RSV08} related to the construction, linearity, kernel and classification of such codes. These codes are also compared with the $\Z_4$-linear, $\Z_8$-linear, and in general $\Z_{2^s}$-linear Hadamard codes with $s\geq 2$.
This paper is organized as follows.
In Section \ref{Sec:GrayMap}, we recall the definition of the Gray map considered in this paper and some of its properties.
In Section \ref{Sec:construction}, we describe the construction of  $\Z_2\Z_4\Z_8$-linear Hadamard codes of type $(\alpha_1,\alpha_2, \alpha_3;t_1,t_2, t_3)$ with $\alpha_1\not =0$, $\alpha_2\not =0$, and $\alpha_3\not =0$.
We see that they are not included neither in the family of  $\Z_4$-linear Hadamard codes, nor in the family of  $\Z_8$-linear Hadamard codes, nor in the family of  $\Z_2\Z_4$-linear Hadamard codes with $\alpha_1 \not =0$. 
Indeed, we see that all the nonlinear $\Z_2\Z_4\Z_8$-linear Hadamard codes of lenght $2^t$ with $\alpha_1\not =0$, $\alpha_2\not =0$, and $\alpha_3\not =0$ are not equivalent to any $\Z_2\Z_4\Z_8$-linear Hadamard code of any other type, any $\Z_2\Z_4$-linear Hadamard code, and  any $\Z_{2^s}$-linear Hadamard code, with $s\geq 2$, of the same length $2^t$.

\section{Preliminary results on the Gray map}\label{Sec:GrayMap}


In this section, we focus on the generalized Gray maps considered in this paper for elements of $\Z_4$ and $\Z_8$, and in general of $\Z_{2^s}$, $s\geq 2$. We include some of its properties used in the paper.


We consider the Carlet's Gray map from $\Z_{2^s}$ to $\Z_2^{2^{s-1}}$ \cite{Carlet} given in (\ref{eq:GrayMapCarlet}).
For $s=2$ and $s=3$, the Gray maps $\phi_2$ and $\phi_3$ considered in the paper for the elements of $\Z_4$ and $\Z_8$, respectively, are the following:
\begin{equation*}
\left.\begin{array}{cccccccccl}
\phi_2: &\Z_4 \longrightarrow \Z_2^2 &&& \phi_3: & \Z_8 \longrightarrow \Z_2^4\\
&0\mapsto(0,0) &&& & 0\mapsto (0,0,0,0) \\
&1\mapsto(0,1)&&& &1\mapsto (0,1,0,1)\\
&2\mapsto(1,1)&&& &2 \mapsto (0,0,1,1)\\
&3\mapsto(1,0)&&& &3 \mapsto (0,1,1,0)\\
 &            &&& &4 \mapsto (1,1,1,1)\\
  &           &&& &5 \mapsto (1,0,1,0)\\
   &          &&& &6 \mapsto (1,1,0,0)\\
    &         &&& &7 \mapsto (1,0,0,1).\\
\end{array}\right.
\end{equation*}

From \cite{HadamardZps}, we have the following results:


\begin{corollary}\cite{HadamardZps}\label{coro4}
 Let $\lambda, \mu \in \Z_2$. Then,  $\phi_s(\lambda \mu 2^{s-1})= \lambda \phi_s(\mu 2^{s-1})=\lambda \mu \phi_s(2^{s-1})$.
\end{corollary}





\begin{corollary}\cite{HadamardZps} \label{lemma2}
Let $u, v$ $\in \Z_{2^s}$. Then, $\phi_s(2^{s-1}u+v)= \phi_s(2^{s-1}u)+\phi_s(v)$.
\end{corollary}

\begin{proposition}\cite{HadamardZps}\label{disweight}
Let $u,v\in\Z_{2^s}$. Then, $$d_H(\phi_s(u),\phi_s(v))=\wt_H(\phi_s(u-v)).$$
\end{proposition}

By Proposition \ref{disweight}, the $\Z_2\Z_4\Z_8$-linear codes obtained from the Gray map $\Phi$ are distance invariant, that is, the Hamming weight distribution is invariant under translation by a codeword. Therefore, their minimum distance coincides with the minimum weight.

\section{Construction of $\Z_2\Z_4\Z_8$-additive Hadamard codes}
\label{Sec:construction}

The description of generator matrices having minimum number of rows for $\Z_4$-additive, $\Z_{2^s}$-additive, and in general for $\Z_{p^s}$-additive Hadamard codes, with $s\geq 2$ and $p$ prime,  are given in \cite{Kro:2001:Z4_Had_Perf}, \cite{KernelZ2s}, and \cite{HadamardZps}, respectively. 
Similarly, generator matrices having minimum number of rows for
$\Z_p\Z_{p^2}$-additive Hadamard codes with $\alpha_1\not =0, \alpha_2\not =0$ and $p$ prime, as long as an iterative construction of these matrices, are given in \cite{PRV06,RSV08} when $p=2$ and in \cite{ZpZp2Construction, WCC2022, ITW2022} when $p\geq 3$. 
In this section, we generalize these results for $\Z_2\Z_4\Z_8$-additive Hadamard codes with $\alpha_1\not =0$, $\alpha_2\not =0$, and $\alpha_3\not =0$. Specifically, we define an iterative construction for the generator matrices of these codes and establish that they generate $\Z_2\Z_4\Z_8$-additive Hadamard codes.

 Let $\zero, \one,\two,\ldots, \mathbf{7}$ be the vectors having the elements $0, 1, 2, \ldots, 7$  repeated in each coordinate, respectively. If $A$ is a generator matrix of a $\Z_2\Z_4\Z_8$-additive code, that is, a subgroup of $\Z_2^{\alpha_1} \times \Z_4^{\alpha_2} \times \Z_8^{\alpha_3}$ for some integers $\alpha_1,\alpha_2,\alpha_3\geq 0$, then we denote by $A_1$ the submatrix of $A$ with the first $\alpha_1$ columns over $\Z_2$, $A_2$ the submatrix with the next $\alpha_2$ columns over $\Z_4$, and $A_3$ the submatrix with the last $\alpha_3$ columns over $\Z_8$. We have that $A=(A_1\mid A_2\mid A_3)$, where the number of columns of $A_i$ is $\alpha_i$ for $i\in \{1,2,3\}$. 

Let $t_1\geq 1$, $t_2\geq 0$, and $t_3\geq 1$ be integers. Now, we construct recursively matrices $A^{t_1,t_2,t_3}$ having $t_1$ rows of order $8$, $t_2$ rows of order $4$, and $t_3$ rows of order $2$ as follows. First, we consider the following matrix:
\begin{equation}\label{eq:recGenMatrix0}
A^{1,0,1}=
\left(\begin{array}{cc|c|c}
1 & 1  & 2  &4 \\
0  & 1 &1  &1  \\
\end{array}\right). 
\end{equation}
Then, we apply the following constructions. 
If we have a matrix  $A^{\ell-1,0,1}=(A_1 \mid A_2 \mid A_3)$, with $\ell \geq 2$, we may construct the matrix
\begin{equation}\label{eq:recGenMatrix1}
\footnotesize
A^{\ell,0,1}=
\left(\begin{array}{cc|ccccc|ccccc}
A_1 & A_1 &M_1 &A_2 &A_2 &A_2 &A_2 &M_2 &A_3 &A_3 &\cdots &A_3 \\
\mathbf{0}  & \mathbf{1} & \mathbf{1}  &\mathbf{0} &\mathbf{1}  &\mathbf{2} &\mathbf{3} &\mathbf{1} &\mathbf{0} &\mathbf{1} &\cdots &\mathbf{7} \\
\end{array}\right),
\end{equation}
where $M_1=\{\mathbf{z}^T:\mathbf{z}\in\lbrace2\rbrace\times \lbrace0,2\rbrace^{\ell-1}\}$ and $M_2=\{\mathbf{z}^T: \mathbf{z}\in\lbrace4\rbrace\times \lbrace0,2,4,6\rbrace^{\ell-1}\}$. We perform construction (\ref{eq:recGenMatrix1}) until $\ell=t_1$. If we have a matrix $A^{t_1,\ell-1,1}=(A_1\mid A_2 \mid A_3)$, with $t_1 \geq 1$ and $\ell\geq 1$, we may construct the matrix 
\begin{equation}\label{eq:recGenMatrix2}
\footnotesize
A^{t_1,\ell,1}=
\left(\begin{array}{cc|ccccc|cccc}
A_1 & A_1 &M_1 &A_2 &A_2 &A_2 &A_2 &A_3 &A_3 &A_3 &A_3 \\
\mathbf{0}  & \mathbf{1} & \mathbf{1}  &\mathbf{0} &\mathbf{1}  &\mathbf{2} &\mathbf{3}  &\mathbf{0} &\mathbf{2} &\mathbf{4} &\mathbf{6} \\
\end{array}\right),
\end{equation}
where $M_1=\{\mathbf{z}^T:\mathbf{z}\in\lbrace2\rbrace\times \lbrace0,2\rbrace^{t_1+\ell-1}\}$. We repeat construction (\ref{eq:recGenMatrix2}) until $\ell=t_2$. Finally, if we have a matrix $A^{t_1,t_2,\ell-1}=(A_1\mid A_2 \mid A_3)$, with $t_1\geq 1$, $t_2\geq 0$, and $\ell\geq 2$, we may construct the matrix
\begin{equation}\label{eq:recGenMatrix3}
A^{t_1,t_2,\ell}=
\left(\begin{array}{cc|cc|cc}
A_1 & A_1 & A_2 & A_2 &A_3 &A_3 \\
\mathbf{0}  & \mathbf{1} & \mathbf{0} &\mathbf{2}  &\mathbf{0}  &\mathbf{4}  \\
\end{array}\right).
\end{equation}
We repeat construction (\ref{eq:recGenMatrix3}) until $\ell=t_3$. Thus, in this way, we obtain $A^{t_1,t_2,t_3}$.

Summarizing, in order to achieve $A^{t_1,t_2,t_3}$ from $A^{1,0,1}$, first we add $t_1-1$ rows of order $8$  by applying construction (\ref{eq:recGenMatrix1}) $t_1-1$ times, starting from $A^{1,0,1}$ up to obtain $A^{t_1,0,1}$; then we add $t_2$ rows of order $4$ by applying construction (\ref{eq:recGenMatrix2}) $t_2$ times, up to generate $A^{t_1,t_2,1}$; and,  finally, we add $t_3-1$ rows of order $2$ by applying construction (\ref{eq:recGenMatrix3}) $t_3-1$ times to achieve $A^{t_1,t_2,t_3}$. Note that in the first row there is always the row $(\one \mid \two \mid \four)$.

\begin{example} \label{ex:MatricesA}
By using the constructions described in (\ref{eq:recGenMatrix1}), (\ref{eq:recGenMatrix2}), and (\ref{eq:recGenMatrix3}), we obtain the following matrices $A^{2,0,1}$, $A^{1,1,1}$ and $A^{1,1,2}$, respectively, starting from the matrix $A^{1,0,1}$ given in (\ref{eq:recGenMatrix0}):
\begin{equation}\label{eq:A201}
A^{2,0,1}=
\left(\begin{array}{cc|cc|cc}
11 &11 &22 &2 2 2 2 &4 4 4 4 &4 4 4 4 4 4 4 4 \\
0  1 &0  1 &0 2 &1 1 1 1 &0 2 4 6 &1 1 1 1 1 1 1 1 \\
0 0 &1 1 &1 1 &0 1 2 3 &1 1 1 1 &0 1 2 3 4 5 6 7\\ 
\end{array}\right),
\end{equation}
\begin{equation}\label{eq:A111}
A^{1,1,1}=
\left(\begin{array}{cc|cc|c}
 11 &11 &22&2222 &4444\\
 01&01 &02&1111 &1111\\
 00&11 &11&0123 &0246\\ 
\end{array}\right), 
\end{equation}
$$
A^{1,1,2}=
\left(\begin{array}{cc|cc|cc}
 1111 & 1111 &222222 &222222 &4444 &4444\\
 0101 & 0101 &021111 &021111 &1111 &1111\\
 0011 & 0011 &110123 &110123 &0246 &0246\\ 
 0000 &1111  &000000 &222222 &0000 &4444\\
\end{array}\right).
$$
\end{example}

In order to obtain $A^{2,1,1}$, we start with $A^{1,0,1}$,  we apply construction (\ref{eq:recGenMatrix1}) to obtain $A^{2,0,1}=(A_1\mid A_2\mid A_3)$ given in (\ref{eq:A201}), and then we apply (\ref{eq:recGenMatrix2}) to obtain 
 \begin{equation*}\label{eq:A_211}
A^{2,1,1}=\left(\begin{array}{cc|@{}cccccc|cccc}
      A_1  &A_1 &&\begin{matrix}
            22 2 2\\
            0 0 2 2\\
            0 2 0 2\\
      \end{matrix} &A_2  &A_2  &A_2  &A_2 &A_3 &A_3 &A_3 &A_3\\
      \mathbf{0} &\mathbf{1} &&\mathbf{1} &\mathbf{0} &\mathbf{1} &\mathbf{2} &\mathbf{3} &\mathbf{0} &\mathbf{2} &\mathbf{4} &\mathbf{6} \\
   \end{array}\right).
\end{equation*}

The $\Z_2\Z_4\Z_8$-additive code generated by $A^{t_1,t_2,t_3}$ is denoted by ${\cH}^{t_1,t_2,t_3}$, and the corresponding $\Z_2\Z_4\Z_8$-linear code $\Phi( {\cH}^{t_1,t_2,t_3})$ by $H^{t_1,t_2, t_3}$.

\begin{lemma}\label{relation:t1a1}
Let $t_1\geq 1$ and $t_2\geq 0$ be integers. Let ${\cH}^{t_1,t_2,1}$ be the $\Z_2\Z_4 \Z_8$-additive code of type $(\alpha_1, \alpha_2, \alpha_3; t_1,t_2,1)$ generated by the matrix $A^{t_1,t_2,1}$. Then, $2^{t_1+t_2}=\alpha_1$, $4^{t_1+t_2}=\alpha_1+2\alpha_2$ and $8^{t_1}4^{t_2}=\alpha_1+2\alpha_2+ 4\alpha_3$.
\end{lemma}
\begin{proof}
First, we prove this lemma for the code ${\cH}^{t_1,0,1}$ by induction on $t_1\geq 1$. Note that the lemma is true for the code ${\cH}^{1,0,1}$ of type $(2,1,1;1,0,1)$. Assume that the lemma is true for the code ${\cH}^{t_1,0,1}$ of type $(\alpha_1, \alpha_2, \alpha_3; t_1,0,1)$, that is, 
\begin{align}\label{t1_0_1}
    2^{t_1}=\alpha_1,
    4^{t_1}=\alpha_1+2\alpha_2~and~
    8^{t_1}=\alpha_1+2\alpha_2+4\alpha_3. 
\end{align}
By using construction (\ref{eq:recGenMatrix1}), the type of ${\cH}^{t_1+1,0,1}$ is $(\alpha'_1, \alpha'_2, \alpha'_3; t_1+1,0,1)$, where 
\begin{align}\label{t1p1_0_1}
     \alpha'_1=2\alpha_1,
    \alpha'_2=2^{t_1}+4\alpha_2~and~
    \alpha'_3=4^{t_1}+8\alpha_3.
\end{align}
Thus, from (\ref{t1_0_1}) and (\ref{t1p1_0_1}), $2^{t_1+1}=2\alpha_1=\alpha'_1$, $4^{t_1+1}=4\alpha_1+8\alpha_2=2\alpha_1+2\alpha_1+8\alpha_2=\alpha'_1+2^{t_1+1}+8\alpha_2=\alpha'_1+2\alpha'_2$ and $8^{t_1+1}=8\alpha_1+16\alpha_2+32\alpha_3=2\alpha_1+(2\alpha_1+8\alpha_2)+(4\alpha_1+8\alpha_2+32\alpha_3)=2\alpha_1+(2^{t_1+1}+8\alpha_2)+(4^{t_1+1}+32\alpha_3)=\alpha'_1+2\alpha'_2+4\alpha'_3$. Therefore, the lemma is true for the code ${\cH}^{t_1,0,1}$.

Next, we prove this lemma for the code ${\cH}^{t_1,t_2,1}$ by induction on $t_2\geq 0$. Assume that the lemma holds for the code ${\cH}^{t_1,t_2,1}$ of type $(\alpha_1, \alpha_2, \alpha_3; t_1,t_2,1)$, that is, 
\begin{align}\label{t1_t2_1}
      2^{t_1+t_2}=\alpha_1,
    4^{t_1+t_2}=\alpha_1+2\alpha_2,~and~
    8^{t_1}4^{t_2}=\alpha_1+2\alpha_2+4\alpha_3.
\end{align}
By using construction (\ref{eq:recGenMatrix2}), the type of ${\cH}^{t_1,t_2+1,1}$ is $(\alpha'_1, \alpha'_2, \alpha'_3; t_1,t_2+1,1)$, where 
\begin{align}\label{t1_t2p1_1}
    \alpha'_1=2\alpha_1,
    \alpha'_2=2^{t_1+t_2}+4\alpha_2 ~and ~
    \alpha'_3=4\alpha_3.
\end{align}
Thus, from (\ref{t1_t2_1}) and (\ref{t1_t2p1_1}), $2^{t_1+(t_2+1)}=2\alpha_1=\alpha'_1$, $4^{t_1+(t_2+1)}=4\alpha_1+8\alpha_2=2\alpha_1+2\alpha_1+8\alpha_2=\alpha'_1+2^{t_1+t_2+1}+8\alpha_2=\alpha'_1+2\alpha'_2$ and $8^{t_1}4^{t_2+1}=4\alpha_1+8\alpha_2+16\alpha_3=2\alpha_1+(2\alpha_1+8\alpha_2)+16\alpha_3=\alpha'_1+(2^{t_1+t_2+1}+8\alpha_2)+4\alpha'_3=\alpha'_1+2\alpha'_2+4\alpha'_3$. Therefore, the lemma is true for the code ${\cH}^{t_1,t_2+1,1}$. This completes the proof.
\end{proof}

\begin{proposition}\label{relation:t1a1a3}
Let $t_1\geq 1$, $t_2\geq 0$, and $t_3\geq 1$ be integers. Let ${\cH}^{t_1,t_2,t_3}$ be the $\Z_2\Z_4 \Z_8$-additive code of type $(\alpha_1, \alpha_2, \alpha_3; t_1,t_2,t_3)$ generated by the matrix $A^{t_1,t_2,t_3}$. Then,  
\begin{align} \label{eq:propt1t2t3}
\begin{split}
    &\alpha_1=2^{t_1+t_2+t_3-1},\\ 
    &\alpha_1+2\alpha_2=4^{t_1+t_2}2^{t_3-1},\\ 
    &\alpha_1+2\alpha_2+ 4\alpha_3= 8^{t_1}4^{t_2}2^{t_3-1}.
\end{split}
\end{align}
\end{proposition}
\begin{proof}
We prove this result for the code ${\cH}^{t_1,t_2,t_3}$ by induction on $t_3\geq 1$. By Lemma \ref{relation:t1a1}, the proposition is true for $t_3=1$, that is, for the code ${\cH}^{t_1,t_2,1}$.
Assume that it holds for the code ${\cH}^{t_1,t_2,t_3}$ of type $(\alpha_1, \alpha_2, \alpha_3; t_1,t_2,t_3)$, that is, (\ref{eq:propt1t2t3}) holds.
By using construction (\ref{eq:recGenMatrix3}), the type of ${\cH}^{t_1,t_2,t_3+1}$ is $(\alpha'_1, \alpha'_2, \alpha'_3; t_1,t_2,t_3+1)$, where 
\begin{align}\label{t1_t2p1_t_3}
    \alpha'_1=2\alpha_1,
    \alpha'_2=2\alpha_2, ~\mbox{and} ~
    \alpha'_3=2\alpha_3.
\end{align}
Thus, from (\ref{eq:propt1t2t3}) and (\ref{t1_t2p1_t_3}), $2^{t_1+t_2+t_3}=2\alpha_1=\alpha'_1$, $4^{t_1+t_2}2^{t_3}=2\alpha_1+4\alpha_2=\alpha'_1+2\alpha'_2$ and $8^{t_1}4^{t_2}2^{t_3}=2\alpha_1+4\alpha_2+8\alpha_3=\alpha'_1+2\alpha'_2+4\alpha'_3$. Therefore, the proposition is true for the code ${\cH}^{t_1,t_2,t_3+1}$. This completes the proof.
\end{proof}

\begin{corollary} \label{coro:Length}
Let $t_1\geq 1$, $t_2\geq 0$, and $t_3\geq 1$ be integers. Let ${\cH}^{t_1,t_2,t_3}$ be the $\Z_2\Z_4 \Z_8$-additive code of type $(\alpha_1, \alpha_2, \alpha_3; t_1,t_2,t_3)$ generated by the matrix $A^{t_1,t_2,t_3}$. 
 Then,
\begin{align*}
\begin{split}
    &\alpha_1=2^{t_1+t_2+t_3-1},\\ 
    &\alpha_2=4^{t_1+t_2}2^{t_3-2}-2^{t_1+t_2+t_3-2},\\ &\alpha_3=8^{t_1}4^{t_2-1}2^{t_3-1}-4^{t_1+t_2-1}2^{t_3-1}.
\end{split}
\end{align*}
\end{corollary}

\begin{remark}
By Corollary \ref{coro:Length},  we have that the $\Z_2\Z_4 \Z_8$-additive codes $\mathcal{H}^{t_1,t_2,t_3}$ of type $(\alpha_1, \alpha_2, \alpha_3; t_1,t_2,t_3)$ generated by the matrix $A^{t_1,t_2,t_3}$, so constructed recursively from (\ref{eq:recGenMatrix1}), (\ref{eq:recGenMatrix2}), and (\ref{eq:recGenMatrix3}), satisfy that $\alpha_1\not =0$, $\alpha_2\not =0$, and $\alpha_3\not =0$. 
\end{remark}

\begin{remark}
We can see the construction of the generator matrices $A^{t_1,t_2,t_3}$ as a generalization of the recursive construction of the generator matrices of the $\Z_2\Z_4$-additive Hadamard codes of type $(\alpha_1,\alpha_2;t_1,t_2)$ with $\alpha_1\not =0$ and $\alpha_2 \not =0$, given in \cite{RSV08}. Note that if we do not consider the coordinates over $\Z_8$ in constructions (\ref{eq:recGenMatrix1}), (\ref{eq:recGenMatrix2}), and (\ref{eq:recGenMatrix3}), we have that (\ref{eq:recGenMatrix1}) and (\ref{eq:recGenMatrix2}) become 
\begin{equation}\label{eq:recZ2Z4GenMatrix1}
\footnotesize
A^{\ell,1}=
\left(\begin{array}{cc|ccccc}
A_1 & A_1 &M_1 &A_2 &A_2 &A_2 &A_2  \\
\mathbf{0}  & \mathbf{1} & \mathbf{1}  &\mathbf{0} &\mathbf{1}  &\mathbf{2} &\mathbf{3} \\
\end{array}\right),
\end{equation}
where $A^{\ell-1,1} =(A_1\mid A_2)$ and  $M_1=2A_1=\{\mathbf{z}^T:\mathbf{z}\in\lbrace2\rbrace\times \lbrace0,2\rbrace^{\ell-1}\}$ (up to a column permutation); and
construction (\ref{eq:recGenMatrix3}) become 
\begin{equation}\label{eq:recZ2Z4GenMatrix3}
A^{t_1,\ell}=
\left(\begin{array}{cc|cc}
A_1 & A_1 & A_2 & A_2 \\
\mathbf{0}  & \mathbf{1} & \mathbf{0} &\mathbf{2}   \\
\end{array}\right),
\end{equation}
where $A^{t_1, \ell-1} =(A_1\mid A_2)$.  Then, starting from the following matrix:
\begin{equation}\label{eq:recZ2Z4GenMatrix0}
A^{1,1}=
\left(\begin{array}{cc|c}
1 & 1  & 2  \\
0  & 1 &1    \\
\end{array}\right), 
\end{equation}
and applying (\ref{eq:recZ2Z4GenMatrix1}) and (\ref{eq:recZ2Z4GenMatrix3}) in the same way as above, we obtain the generator matrices $A^{t_1,t_2}$ of the known $\Z_2\Z_4$-additive Hadamard codes of type $(\alpha_1,\alpha_2;t_1,t_2)$ with $\alpha_1\not =0$ and $\alpha_2 \not =0$ \cite{PRV06,RSV08}. The $\Z_2\Z_4$-additive code generated by $A^{t_1,t_2}$ is denoted by ${\cH}^{t_1,t_2}$, and the corresponding $\Z_2\Z_4$-linear code $\Phi( {\cH}^{t_1,t_2})$ by $H^{t_1,t_2}$.
\end{remark}

When we include all the elements of  $\Z_{2^i}$, where $1\leq i\leq 3$,  as coordinates of a vector, we place them in increasing order. 
For a set $S \subseteq \Z_{2^i}$ and $\lambda \in \Z_{2^i}$, where $i\in \{1,2,3\}$, we define $\lambda S=\{\lambda j: j\in S\}$ and $S+\lambda=\{j+\lambda: j\in S\}$. As before, when including all the elements in those sets as coordinates of a vector, we place them in increasing order. 
For example, 
$2\Z_8=\{0,4,6,8\}$, $(\Z_4,\Z_4)=(0,1,2,3,0,1,2,3) \in \Z_4^8$ and $(\Z_2 \mid \Z_4\mid 2\Z_8, 4\Z_8)=(0,1 \mid 0,1,2,3\mid 0,2,4,6, 0,4) \in \Z_2^2 \times \Z_4^4 \times \Z_8^6$.

\begin{lemma}\label{Lemma:cases}
Let $1\leq i\leq 3$ and  $j \in \{0, 1,\dots, i-1\}$. 
\begin{enumerate} 
\item \label{2z22} If $\mu \in 2^j\Z_{2^i}$, then $2^j\Z_{2^i}+\mu = 2^j\Z_{2^i}$.
\item\label{item:vector} If $\mu \in 2^j\Z_{2^i}$, then $(2^j\Z_{2^i}, \stackrel{m}{\dots}, 2^j\Z_{2^i})+ \mu \one$, where $m \geq  1$, is a permutation of the vector $(2^j\Z_{2^i}, \stackrel{m}{\dots}, 2^j\Z_{2^i})$.
\item \label{2z2_}If $\mu \in 2\Z_{2^i}$, then $(\Z_{2^i}{\setminus}2\Z_{2^i})+\mu = \Z_{2^i}{\setminus}2\Z_{2^i}$.
\item\label{item:N2^-} If $\mu \in\Z_{2^i}$, then
$(\zero, \dots, \mathbf{2^i-1})+(\mu,\stackrel{\ell\cdot 2^i}{\dots}, \mu)$, where $\ell\geq 1$ and  $\mathbf{k}=(k,\stackrel{\ell}{\dots}, k)$ for $k \in \Z_{2^i}$, is a permutation of  $(\Z_{2^i},\stackrel{\ell}{\dots}, \Z_{2^i})$.
\end{enumerate}
\end{lemma}
\begin{proof}
Item 1 follows from the fact that $\Z_{2^i}$ is a ring and $2^j\Z_{2^i}$ is an ideal of $\Z_{2^i}$. Item 2 follows from Item 1. 


For Item 3, it $x\in (\Z_{2^i}{\setminus}2\Z_{2^i})+\mu$, then $x-\mu \in \Z_{2^i}{\setminus} 2\Z_{2^i}$. Assume that $x \notin \Z_{2^i}{\setminus} 2\Z_{2^i}$, so $x \in 2\Z_{2^i}$. Since $2\Z_{2^i}$ is an ideal of $\Z_{2^i}$, we have that $x-\mu\in 2\Z_{2^i}$, which is a contradiction. Thus, $x \in \Z_{2^i}{\setminus} 2\Z_{2^i}$ and hence $(\Z_{2^i}{\setminus} 2\Z_{2^i})+\mu \subseteq \Z_{2^i}{\setminus} 2\Z_{2^i}$. In the same way,  $(\Z_{2^i}{\setminus} 2\Z_{2^i})-\mu \subseteq \Z_{2^i}{\setminus} 2\Z_{2^i}$. Hence, $\Z_{2^i}{\setminus} 2\Z_{2^i} \subseteq (\Z_{2^i}{\setminus} 2\Z_{2^i})+\mu $ and therefore
$(\Z_{2^i}{\setminus} 2\Z_{2^i})+\mu = \Z_{2^i}{\setminus}2\Z_{2^i}$.

For Item 4, note that $(\zero, \dots, \mathbf{2^i-1})+(\mu,\stackrel{\ell\cdot 2^i}{\dots}, \mu)$ is a permutation of 
\begin{equation}\label{a}
   (\Z_{2^i},\stackrel{\ell}{\dots}, \Z_{2^i})+(\mu,\stackrel{\ell\cdot 2^i}{\dots}, \mu).
\end{equation}
Since $\Z_{2^i}+\mu =\Z_{2^i}$, (\ref{a}) is a permutation of $(\Z_{2^i},\stackrel{\ell}{\dots}, \Z_{2^i})$.
\end{proof}

\begin{lemma}\label{lemm:uu+lambda}
Let $1\leq i\leq 3$, $\lambda\in\Z_{2^i}{\setminus} 2\Z_{2^i}$, and $u\in\Z_{2^i}^n$. Then,
$$
(u,\stackrel{2^i}{\dots},u)+\lambda(\zero, \dots, \mathbf{2^i-1})
$$
is a permutation of $(\Z_{2^i},\stackrel{n}{\dots}, \Z_{2^i})$.
\end{lemma}
\begin{proof}
Since  $\lambda\in\Z_{2^i}{\setminus} 2\Z_{2^i}$, $\lambda(\zero, \dots, \mathbf{2^i-1})$ is a permutation of $(\zero, \dots, \mathbf{2^i-1})$ and we may consider $\lambda=1$. 
Then, $(u,\dots,u)+(\zero, \dots, \mathbf{2^i-1})$ is a permutation of $(u_1+\Z_{2^i},\dots,u_n+\Z_{2^i})=(\Z_{2^i},\stackrel{n}{\dots}, \Z_{2^i})$, where  $u=(u_1,\dots,u_n)$.
\end{proof}

\begin{lemma}\label{lemm:ord8_u2}
Let  $u=(\mu, \stackrel{m}{\dots}, \mu, 2\Z_{4}, \stackrel{n}{\dots}, 2\Z_{4}, \Z_{4}{\setminus}2\Z_{4},\stackrel{r}{\dots},\Z_{4}{\setminus}2\Z_{4})\in \Z_4^{m+2n+2r}$, where  $m,n,r\geq 0$ and $\mu \in \Z_4{\setminus} 2\Z_4 =\{1,3\}$. Then, 
 $$
 (u,u,u,u)+(\zero, \two, \zero, \two)
 $$
 is a permutation of $(2\Z_{4}, \stackrel{4n}{\dots}, 2\Z_{4}, \Z_{4}{\setminus}2\Z_{4},\stackrel{4r+2m}{\dots},\Z_{4}{\setminus}2\Z_{4})$.
\end{lemma}
\begin{proof}
 By Items \ref{2z22} and \ref{2z2_} of Lemma \ref{Lemma:cases}, $u+\bf{2}$ is a permutation of $(\mu+2, \stackrel{m}{\dots}, \mu+2, 2\Z_{4}, \stackrel{n}{\dots}, 2\Z_{4}, \Z_{4}{\setminus}2\Z_{4},\stackrel{r}{\dots},\Z_{4}{\setminus}2\Z_{4})$. Let ${\bf k}=(\mu, \stackrel{m}{\dots}, \mu)$. Since $\mu\in\{1,3\}$, we have that
$
({\bf k},{\bf k},{\bf k},{\bf k})+(\zero,\two,\zero,\two)
$
is a permutation of $(\Z_{4}{\setminus}2\Z_{4},\stackrel{2m}{\dots},\Z_{4}{\setminus}2\Z_{4})$. Therefore, $(u,u,u,u)+(\zero, \two, \zero, \two)$ is a permutation of $(2\Z_{4}, \stackrel{4n}{\dots}, 2\Z_{4}, \Z_{4}{\setminus}2\Z_{4},\stackrel{4r+2m}{\dots},\Z_{4}{\setminus}2\Z_{4})$.
\end{proof}

\begin{lemma}\label{lemm:ord8_u3} 
Let  $u=(\mu', \stackrel{m'}{\dots}, \mu', \mu'', \stackrel{m'}{\dots}, \mu'', 2\Z_{8}, \stackrel{n'}{\dots}, 2\Z_{8}, \Z_{8}{\setminus}2\Z_{8},\stackrel{r'}{\dots},\Z_{8}{\setminus}2\Z_{8})\in \Z_8^{2m'+4n'+4r'}$, where  $m', n',r'\geq 0$ and $\mu, \mu'\in \Z_8{\setminus}2\Z_8= \{1,3,5,7\}$. Then, 
\begin{enumerate}
    \item\label{ord8u3_case1} $(u,u,u,u)+(\zero, \two,\four, \mathbf{6})$ is a permutation of $(2\Z_{8}, \stackrel{4n'}{\dots}, 2\Z_{8}, \Z_{8}{\setminus}2\Z_{8},\stackrel{4r'+2m'}{\dots},\Z_{8}{\setminus}2\Z_{8})$;
    \item \label{ord8u3_case2}  $(u,u,u,u)+(\zero, \four, \zero, \four)$ is a permutation of $(\mu',\stackrel{4m'}{\dots}, \mu', \mu'+4,\stackrel{4m'}{\dots},\mu'+4, 2\Z_{8}, \stackrel{4n'}{\dots}, 2\Z_{8}, \Z_{8}{\setminus}2\Z_{8},\stackrel{4r'}{\dots},\Z_{8}{\setminus}2\Z_{8})$ if $\mu' =\mu''$ or $\mu'=\mu''+4$, or a permutation of $(2\Z_{8}, \stackrel{4n'}{\dots}, 2\Z_{8},\Z_{8}{\setminus}2\Z_{8},\stackrel{4r'+2m'}{\dots},\Z_{8}{\setminus}2\Z_{8})$ otherwise.
\end{enumerate}
\end{lemma}
\begin{proof}
For Item 1, by Items \ref{2z22} and  \ref{2z2_} of Lemma \ref{Lemma:cases}, if $j\in\{0,2,4,6\}$, then $u+\bf{j}$ is a permutation of $(\mu'+j, \stackrel{m'}{\dots}, \mu'+j, \mu''+j, \stackrel{m'}{\dots}, \mu''+j, 2\Z_{8}, \stackrel{n'}{\dots}, 2\Z_{8}, \Z_{8}{\setminus}2\Z_{8},\stackrel{r'}{\dots},\Z_{8}{\setminus}2\Z_{8})$. Let ${\bf k}'=(\mu', \stackrel{m'}{\dots}, \mu', \mu'' \stackrel{m'}{\dots}, \mu'')$.  Since $\mu',\mu''\in\{1,3,5,7\}$,  we have that
$
({\bf k}', \stackrel{4}{\dots}, {\bf k}')+(\zero, \two, \four, \mathbf{6})
$ 
is a permutation of $(\Z_{8}{\setminus}2\Z_{8},\stackrel{2m'}{\dots},\Z_{8}{\setminus}2\Z_{8})$ and hence $(u,u,u,u)+(\zero, \two,\four, \mathbf{6})$ is a permutation of $(2\Z_{8}, \stackrel{4n'}{\dots}, 2\Z_{8}, \Z_{8}{\setminus}2\Z_{8},\stackrel{4r'+2m'}{\dots},\Z_{8}{\setminus}2\Z_{8})$. 

For item 2, we have that
$
({\bf k}', \stackrel{4}{\dots}, {\bf k}')+(\zero, \four, \zero, \four)
$
is a permutation of $(\mu',\stackrel{4m'}{\dots}, \mu', \mu'+4,\stackrel{4m'}{\dots},\mu'+4)$ if $\mu' =\mu''$ or $\mu'=\mu''+4$, or a permutation of $(\Z_{8}{\setminus}2\Z_{8},\stackrel{2m'}{\dots},\Z_{8}{\setminus}2\Z_{8})$ otherwise. Therefore,  $(u,u,u,u)+(\zero, \four, \zero, \four)$ is a permutation of $(\mu',\stackrel{4m'}{\dots}, \mu', \mu'+4,\stackrel{4m'}{\dots},\mu'+4, 2\Z_{8}, \stackrel{4n'}{\dots}, 2\Z_{8}, \Z_{8}{\setminus}2\Z_{8},\stackrel{4r'}{\dots},\Z_{8}{\setminus}2\Z_{8})$ if $\mu' =\mu''$ or $\mu'=\mu''+4$, or a permutation of $(2\Z_{8}, \stackrel{4n'}{\dots}, 2\Z_{8},\Z_{8}{\setminus}2\Z_{8},\stackrel{4r'+2m'}{\dots},\Z_{8}{\setminus}2\Z_{8})$ otherwise.
\end{proof}

\begin{lemma}\label{lemm:ord4_u3}
Let $u=(\mu,\stackrel{m}{\dots}, \mu,4\Z_8,\stackrel{n}{\dots},4\Z_8,2\Z_8{\setminus}4\Z_8, \stackrel{r}{\dots}, 2\Z_8{\setminus}4\Z_8)\in \Z_8^{m+2n+2r}$, where  $m, n,r\geq 0$ and $\mu \in 2\Z_8{\setminus}4\Z_8=\{2,6\}$.
Then,
\begin{enumerate}
    \item\label{lemm:ord4_case1} $(u,u,u,u)+(\zero, \two,\four, \mathbf{6})$ is a permutation of $(2\Z_8, \stackrel{2r+2n+m}{\dots},2\Z_8)$;
    \item \label{lemm:ord4_case2} $(u,u,u,u)+(\zero, \four, \zero, \four)$ is a permutation of $(4\Z_8,\stackrel{4n}{\dots},4\Z_8,2\Z_8{\setminus}4\Z_8, \stackrel{4r+2m}{\dots}, 2\Z_8{\setminus}4\Z_8)$.
\end{enumerate}
\end{lemma}
\begin{proof}
By Item \ref{2z22} of Lemma \ref{Lemma:cases}, if $j\in \{0,4\}$, then $u+\mathbf{j}$ is a permutation of $ (\mu+j,\stackrel{m}{\dots}, \mu+j,4\Z_8,\stackrel{n}{\dots},4\Z_8,2\Z_8{\setminus}4\Z_8, \stackrel{r}{\dots}, 2\Z_8{\setminus}4\Z_8)$. Similarly,  if $j\in \{2,6\}$, then $u+\mathbf{j}$ is a permutation of $ (\mu+j,\stackrel{m}{\dots}, \mu+j,4\Z_8,\stackrel{r}{\dots},4\Z_8,2\Z_8{\setminus}4\Z_8, \stackrel{n}{\dots}, 2\Z_8{\setminus}4\Z_8)$. Let $\mathbf{k}=(\mu, \stackrel{m}{\dots}, \mu)$.  

For Item 1, since $\mu \in \{2,6\}$, we have that 
    $
    (\mathbf{k}, \stackrel{4}{\dots}, \mathbf{k})+(\zero, \two, \four, \mathbf{6})
    $
    is a permutation of $(2\Z_8,\stackrel{m}{\dots}, 2\Z_8)$, and hence $(u,u,u,u)+(\zero, \two,\four, \mathbf{6})$ is a permutation of $(2\Z_8, \stackrel{2r+2n+m}{\dots},2\Z_8)$. 
    
    For Item 2,  we have that 
    $
    (\mathbf{k}, \stackrel{4}{\dots}, \mathbf{k})+(\zero, \four, \zero, \four)
    $
    is a permutation of $(2\Z_8{\setminus}4\Z_8, \stackrel{2m}{\dots}, 2\Z_8{\setminus}4\Z_8)$. Therefore, $(u,u,u,u)+(\zero, \four, \zero, \four)$ is a permutation of $(4\Z_8,\stackrel{4n}{\dots},4\Z_8,2\Z_8{\setminus}4\Z_8, \stackrel{4r+2m}{\dots}, 2\Z_8{\setminus}4\Z_8)$.
\end{proof}

\medskip
 Let $t_1\geq1, t_2\geq 0$, and $t_3\geq1$ be integers.  Let ${\cal{G}}^{t_1,t_2,t_3}$ be the set of all codewords of the code generated by the matrix obtained from $A^{t_1,t_2,t_3}$ after removing the row $(\mathbf{1\mid 2\mid 4})$.


\begin{lemma}\label{lem:z_t1p1_0_1}
    Let $t_1\geq 1$ be an integer. Let $$\zz=(u_1,u_1 \mid x_1,u_2, u_2, u_2, u_2\mid x_2,u_3,\stackrel{8}{\dots},u_3)\in {\cal{G}}^{t_1+1,0,1},$$  where $\uu=(u_1\mid u_2\mid u_3) \in {\cal{G}}^{t_1,0,1}$ and $x_{i-1}\in (2\Z_{2^i})^{2^{(i-1)t_1}}$ for $i \in \{2,3\}$.
    Then,
    \begin{enumerate}
    \item if $o(\zz)=8$, then $x_{i-1}$ is a permutation  of $(2\Z_{2^{i}},\stackrel{2^{(i-1)(t_1-1)}}{\dots}, 2\Z_{2^{i}})$ for $i \in \{2,3\}$.
        \item if $o(\zz)=4$, then $x_1=\zero$ and $x_2$ is a permutation of $(4\Z_8, \stackrel{2\cdot 4^{t_1-1}}{\dots}, 4\Z_8)$.
        \item if $o(\zz)=2$, then $x_1=\zero$ and $x_2=\zero$.
    \end{enumerate}
\end{lemma}
\begin{proof}
    Let $\ww_j$, where $j\in \{1,\dots,t_1+2\}$, be the $j$th row of the matrix $A^{t_1+1,0,1}$. Note that $\ww_1=(\mathbf{1\mid 2\mid 4})$, and $\ww_2,\dots, \ww_{t_1+2}$ are the rows of order 8, where $\ww_{t_1+2}=(\zero, \one \mid \one,\zero, \one, \two,  \mathbf{3}\mid \one, \zero, \dots, \mathbf{7})$.  
    Since any element of ${\cal{G}}^{t_1+1,0,1}$ can be written as $\zz+\lambda \ww_{t_1+2}$, where $\lambda \in \Z_8$, then $\zz=\sum_{j=2}^{t_1+1}r_j\ww_j=(u_1,u_1 \mid x_1,u_2, u_2, u_2, u_2\mid x_2,u_3,\stackrel{8}{\dots},u_3)$, where $r_j \in \Z_8$.  By construction, 
    $x_1$ and $x_2$ are generated by the rows of $M'_1=\{\mathbf{z}^T:\mathbf{z}\in \lbrace 0,2\rbrace^{t_1}\}$ and $M'_2=\{\mathbf{z}^T: \mathbf{z}\in \lbrace0,2,4,6\rbrace^{t_1}\}$, respectively. Thus,  $x_1=\zero$ or $x_1$ is a permutation of $(2\Z_4, \stackrel{2^{t_1-1}}{\dots}, 2\Z_4)$, and $x_2=\zero$ or $x_2$ is a permutation of $(2\Z_8, \stackrel{4^{t_1-1}}{\dots}, 2\Z_8)$ or $(4\Z_8, \stackrel{2\cdot 4^{t_1-1}}{\dots}, 4\Z_8)$.
    
    For Item 1, we have that there exists at least one $j\in \{2,\dots, t_1+1\}$ such that $r_j \in \{1,3,5,7\}$. Therefore, by  Item \ref{2z22} of Lemma \ref{Lemma:cases}, $x_{i-1}$ is a permutation  of $(2\Z_{2^{i}},\stackrel{2^{(i-1)(t_1-1)}}{\dots}, 2\Z_{2^{i}})$ for $i \in \{2,3\}$. 

    For Item 2, we have that $r_j\in 2\Z_8$ for all $j\in \{2,\dots, t_1+1\}$ and there exist at least one $j\in \{2,\dots, t_1+1\}$ such that $r_j\in \{2,6\}$. Therefore, $x_1=\zero$ and, by Item \ref{2z22} of Lemma \ref{Lemma:cases}, $x_2$ is a permutation of $(4\Z_8, \stackrel{2\cdot 4^{t_1-1}}{\dots}, 4\Z_8)$. 

    For Item 3,  we have that $r_j\in 4\Z_8$ for all $j\in \{2,\dots, t_1+1\}$ and there exist at least one $j\in \{2,\dots, t_1+1\}$ such that  $r_j=4$. Therefore, $x_1=\zero$ and $x_2=\zero$.
\end{proof}

\begin{lemma}\label{lem:z_t1t2p1_1}
     Let $t_1\geq 1$ and $t_2\geq 0$ be integers.   Let $$\zz=(u_1,u_1 \mid x_1,u_2, u_2, u_2, u_2\mid u_3,u_3, u_3, u_3) \in {\cal{G}}^{t_1,t_2+1,1},$$ where $\uu=(u_1\mid u_2\mid u_3) \in {\cal{G}}^{t_1,t_2,1}$ and $x_1\in (2\Z_4)^{2^{t_1+t_2}}$.
    Then,
    \begin{enumerate}
    \item if $o(\zz)=8$, then $x_1$ is a permutation  of $(2\Z_4,\stackrel{2^{t_1+t_2-1}}{\dots}, 2\Z_4)$.
        \item if $o(\zz)=4$, then $x_1=\zero$ if $u_1=\zero$, and $x_1$ is a permutation of $(2\Z_4,\stackrel{2^{t_1+t_2-1}}{\dots}, 2\Z_4)$ otherwise.
        \item if $o(\zz)=2$, then $x_1=\zero$.
    \end{enumerate}
\end{lemma}
\begin{proof}
    Let $\ww_i$, where $i\in \{1,\dots,t_1+t_2+2\}$, be the $i$th row of the matrix $A^{t_1, t_2+1,1}$. Note that $\ww_1=(\mathbf{1\mid 2\mid 4})$, $\ww_2, \dots, \ww_{t_1+1}$ are the rows of order $8$, and $\ww_{t_1+2}, \dots, \ww_{t_1+t_2+2}$ are the rows of order $4$, where $\ww_{t_1+t_2+2}=(\zero, \one \mid \one, \zero, \one, \two, \mathbf{3}\mid  \zero, \two, \four, \mathbf{6})$.  
    Since any element of ${\cal{G}}^{t_1,t_2+1,1}$ can be written as $\zz+\lambda \ww_{t_1+t_2+2}$, where $\lambda \in \{0,1,2,3\}$, then $\zz=\sum_{i=2}^{t_1+t_2+1}r_i\ww_i=(u_1,u_1 \mid x_1,u_2, u_2, u_2, u_2\mid u_3,u_3, u_3, u_3)$, where $r_i \in \Z_8$ for $i \in \{2,\dots, t_1+1\}$ and $r_i\in \{0,1,2,3\}$ for $i\in \{t_1+2,\dots, t_1+t_2+1\}$. By construction, 
    $x_1$ is generated by the rows of $M'_1=\{\mathbf{z}^T:\mathbf{z}\in \lbrace 0,2\rbrace^{t_1+t_2}\}$. Thus,  $x_1=\zero$ or $x_1$ is a permutation of $(2\Z_4, \stackrel{2^{t_1+t_2-1}}{\dots}, 2\Z_4)$.

    For Item 1, we have that there exists at least one $i\in \{2,\dots, t_1+1\}$ such that $r_i \in \{1,3,5,7\}$. Therefore, since $x_1$ is of order at most two,  $x_1\neq \zero$.

 For Item 2, we have that $r_i\in 2\Z_8$ for all $i \in \{2,\dots, t_1+1\}$ and $r_i\in \{0,1,2,3\}$ for all $i\in \{t_1+2,\dots, t_1+t_2+1\}$.  Note that, since $x_1$ and $u_1$ are of order at most two,  $x_1\neq \zero$ if and only if there exists at least one $i$ for $i\in \{t_1+2,\dots, t_1+t_2+1\}$ such that $r_i\in \{1,3\}$, or equivalently, if and only if $u_1\neq \zero$.


For Item 2, we have that $r_i\in 4\Z_8=\{0,4\}$ for all $i \in \{2,\dots, t_1+1\}$ and $r_i\in \{0,2\}$ for all $i\in \{t_1+2,\dots, t_1+t_2+1\}$. Therefore, since $x_1$ is of order at most two, $x_1=\zero$.
\end{proof}

\begin{lemma}\label{lemma:same_num1}
Let $t_1\geq 1$ be an integer. Let ${\cH}^{t_1,0,1}$ be the $\Z_2\Z_4 \Z_8$-additive code of type $(\alpha_1,\alpha_2,$ $\alpha_3;$ $t_1,0,1)$ generated by the matrix $A^{t_1,0,1}$. 
Let $\uu=(u_1\mid u_2\mid u_3) \in {\cal{G}}^{t_1,0,1}$. Then,

 \begin{enumerate}
 \item \label{proper-1}if $o(\uu)=8$, then $u_1$ contains every element of $\Z_2$ the same number of times, $u_2$ is a permutation of $(\mu, \stackrel{m}{\dots}, \mu, 2\Z_{4}, \stackrel{n}{\dots}, 2\Z_{4}, \Z_{4}{\setminus}2\Z_{4},\stackrel{r}{\dots},\Z_{4}{\setminus}2\Z_{4})$ for some integers $m,n,r\geq 0$ and $\mu \in \{1,3\}$, and $u_3$ is a permutation of $(\mu', \stackrel{m'}{\dots}, \mu', \mu'', \stackrel{m'}{\dots}, \mu'', 2\Z_{8}, \stackrel{n'}{\dots}, 2\Z_{8}, \Z_{8}{\setminus}2\Z_{8},\stackrel{r'}{\dots},\Z_{8}{\setminus}2\Z_{8})$ for some integers  $m', n',r'\geq 0$ and $\mu, \mu'\in \{1,3,5,7\}$.
 
  \item \label{proper-2} if $o(\uu)=4$, then $u_1=\zero$,
         $u_2$ contains the element in $2\Z_4{\setminus}\{0\}=\{2\}$ exactly $\frac{1}{2}(\frac{\alpha_1}{2}+\alpha_2)=4^{t_1-1}$ times and  $\frac{\alpha_2}{2}-\frac{\alpha_1}{4}=4^{t_1-1} -2^{t_1-1}$ times the element  $0$,
   and  $u_3$ is a permutation of $(\mu,\stackrel{m}{\dots}, \mu,4\Z_8,\stackrel{n}{\dots},4\Z_8,2\Z_8{\setminus}4\Z_8, \stackrel{r}{\dots}, 2\Z_8{\setminus}4\Z_8)$ for some integers  $m, n,r\geq 0$ and $\mu \in \{2,6\}$.

  \item \label{proper-3} if $o(\uu)=2$, then $u_1=\zero$, $u_2=\zero$, and $u_3$ contains the element in $4\Z_8{\setminus}\{0\}=\{4\}$ exactly $\frac{1}{4}(\frac{\alpha_1}{2}+\alpha_2+2\alpha_3)= 8^{t_1-1} $ times and  $\frac{\alpha_3}{2}-\frac{1}{4}(\frac{\alpha_1}{2}+\alpha_2)= 8^{t_1-1} -4^{t_1-1}$ times the element  $0$.
   \end{enumerate}
\end{lemma}
\begin{proof}
We prove this lemma by induction on $t_1\geq 1$. If $t_1=1$,  then by Lemma \ref{relation:t1a1}, $\alpha_1=2$, $\alpha_2=1$, $\alpha_3=1$, and ${\cal{G}}^{1,0,1}= \langle (0,1\mid 1\mid 1)\rangle$. Let $\uu=(u_1\mid u_2\mid u_3) \in {\cal{G}}^{1,0,1}$. Then, $\uu=\lambda(0,1\mid 1\mid 1)$, where  $ \lambda \in \Z_8$. Thus, we have that $u_1=\lambda(0,1)$, $u_2=(\lambda)$, and $u_3=(\lambda)$. If $o(\uu)=8$, then $\lambda \in \Z_8{\setminus}2\Z_8$. Therefore, $\uu$ satisfies property \ref{proper-1}. If $o(\uu)=4$, then $\lambda \in \{2,6\}$. In this case, $u_1=(0,0)$, $u_2=(2)$ contains the element in $2\Z_4{\setminus}\{0\}=\{2\}$ exactly $1=\frac{1}{2}(\frac{\alpha_1}{2}+\alpha_2)$ time and  $0=\frac{\alpha_2}{2}-\frac{\alpha_1}{4}$ times the element  $0$, and $u_3=(\lambda)$. Thus, $\uu$ satisfies property \ref{proper-2}. If $o(\uu)=2$, then $\lambda=4$. In this case, $u_1=(0,0)$, $u_2=(0)$, and $u_3=(4)$ contains the element in $4\Z_8{\setminus}\{0\}=\{4\}$ exactly $1=\frac{1}{4}(\frac{\alpha_1}{2}+\alpha_2+2\alpha_3)$ time and  $0=\frac{\alpha_3}{2}-\frac{1}{4}(\frac{\alpha_1}{2}+\alpha_2)$ times the element  $0$. Thus, $\uu$ satisfies property \ref{proper-3}. Therefore, the lemma is true for $t_1=1$.
 
 Assume that the lemma holds for the code ${\cH}^{t_1,0,1}$ of type $(\alpha_1, \alpha_2, \alpha_3; t_1,0,1)$ with $t_1\geq 1$.
By Lemma \ref{relation:t1a1}, we have that
 \begin{align}\label{eq:t1_t2_1}
      2^{t_1}=\alpha_1,
    4^{t_1}=\alpha_1+2\alpha_2,~\mbox{and}~
    8^{t_1}=\alpha_1+2\alpha_2+4\alpha_3.
\end{align}
 Now, we have to show that the lemma is also true for the code ${\cH}^{t_1+1,0,1}$.
 
 Let $\vv=(v_1\mid v_2 \mid v_3)\in {\cal{G}}^{t_1+1,0,1}$. We can write 
 \begin{align*}
     \vv=\zz+\lambda \ww,
 \end{align*}
 where  $\zz=(u_1,u_1 \mid x_1,u_2, u_2, u_2, u_2\mid x_2,u_3,\stackrel{8}{\dots},u_3)$, $\ww=(\zero, \one \mid \one,\zero, \one, \two,  \mathbf{3}\mid \one, \zero, \dots, \mathbf{7})$, $\uu=(u_1\mid u_2\mid u_3) \in {\cal{G}}^{t_1,0,1}$, $\lambda \in \Z_8$,   $x_1\in (2\Z_4)^{2^{t_1}}$ such that either $x_1=\zero$ or $x_1$ is a permutation of $(2\Z_4, \stackrel{2^{t_1-1}}{\dots}, 2\Z_4)$, and  $x_2\in (2\Z_8)^{4^{t_1}}$ such that either $x_2=\zero$ or $x_2$ is a permutation of $(2\Z_8, \stackrel{4^{t_1-1}}{\dots}, 2\Z_8)$ or $(4\Z_8, \stackrel{2\cdot 4^{t_1-1}}{\dots}, 4\Z_8)$. Then, $v_1=(u_1, u_1)+\lambda(\zero, \one)$ and,  for $i \in \{2,3\}$, 
\begin{align}\label{eq:vi}
     v_i=(x_{i-1}, u_i,\stackrel{2^i}{\dots}, u_i)+\lambda (\one, \zero, \dots, \mathbf{2^i-1}).
\end{align}
If $\mathbf{z}=\zero$, then $\vv=\lambda \ww$ and it is easy to see that $\vv$ satisfies property \ref{proper-1} if $\lambda\in \Z_8{\setminus} 2\Z_8=\{1,3,5,7\}$, property \ref{proper-2} if $\lambda\in\{2,6\}$, and property \ref{proper-3} if $\lambda=4$. 
Therefore, we focus on the case when $\mathbf{z}\neq \zero$. 

Case 1: Assume that $o(\vv)=8$. We have two subcases: when $o(\zz)$  is arbitrary and $\lambda \in \Z_8{\setminus}2\Z_8$, and  when  $o(\zz)=8$ and $\lambda \in 2\Z_8$. In both subcases, note that $v_1$ contains every element of $\Z_2$ the same number of times.  For the first subcase, we have that  $(u_i,\stackrel{2^i}{\dots}, u_i)+\lambda (\zero, \dots, \mathbf{2^i-1})$, for $i \in \{2,3\}$, is a permutation of $(\Z_{2^i},\stackrel{\alpha_i}{\dots},\Z_{2^i})$ by Lemma \ref{lemm:uu+lambda}. 
Thus, from (\ref{eq:vi}), $v_i$ is  a permutation of $(x_{i-1}+\lambda \one,\Z_{2^i},\stackrel{\alpha_i}{\dots}, \Z_{2^i})$. Since either $x_{i-1}+\lambda \one=\lambda \one$, or $x_{i-1}+\lambda \one$ is a permutation of $(\Z_{2^i}{\setminus}2\Z_{2^i},  \stackrel{2^{(i-1)(t_1-1)}}{\dots}, \Z_{2^i}{\setminus}2\Z_{2^i})$, $\vv$ satisfies property \ref{proper-1}.

For the second subcase when $o(\vv)=8$, that is, when $o(\zz)=8$ and $\lambda\in2\Z_8$, we have that $o(\uu)=8$ and, by Item 1 of Lemma \ref{lem:z_t1p1_0_1}, $x_{i-1}$ is a permutation  of $(2\Z_{2^{i}},\stackrel{2^{(i-1)(t_1-1)}}{\dots}, 2\Z_{2^{i}})$ for $i \in \{2,3\}$. 
By induction hypothesis, $\uu$ satisfies property \ref{proper-1} and then $u_2$ is a permutation of 
\begin{equation*}\label{eq:v2}
(\mu, \stackrel{m}{\dots}, \mu, 2\Z_{4}, \stackrel{n}{\dots}, 2\Z_{4}, \Z_{4}{\setminus}2\Z_{4},\stackrel{r}{\dots},\Z_{4}{\setminus}2\Z_{4}),
\end{equation*} where  $m,n,r\geq 0$ and $\mu \in \{1,3\}$, and $u_3$ is a permutation of 
\begin{equation*}\label{eq:v3}
(\mu', \stackrel{m'}{\dots}, \mu', \mu'', \stackrel{m'}{\dots}, \mu'', 2\Z_{8}, \stackrel{n'}{\dots}, 2\Z_{8}, \Z_{8}{\setminus}2\Z_{8},\stackrel{r'}{\dots},\Z_{8}{\setminus}2\Z_{8}),
\end{equation*}
where  $m', n',r'\geq 0$ and $\mu', \mu''\in \{1,3,5,7\}$. From (\ref{eq:vi}), $v_2=(x_1, u_2,u_2,u_2,u_2)+\lambda(\one,\zero,\one,\two,\mathbf{3}).$
If $\lambda \in \{ 0,4\}$, then $v_2=(x_1, u_2,u_2,u_2,u_2)$  in $\vv$ satisfies the same property as $u_2$ in $\uu$; that is, property \ref{proper-1}. If $\lambda\in\{2,6\}$, then $v_2=(x_1,u_2,u_2,u_2,u_2)+(\two,\zero, \two, \zero, \two)$.
By Item \ref{2z22} of Lemma \ref{Lemma:cases}, we have that $x_1+\two$ is a permutation of $(2\Z_{4}, \stackrel{2^{t_1-1}}{\dots},2\Z_{4})$.
Thus, by Lemma \ref{lemm:ord8_u2}, $v_2$ is a permutation of 
$$
(2\Z_{4}, \stackrel{4n+2^{t_1-1}}{\dots},2\Z_{4}, \Z_{4}{\setminus}2\Z_{4},\stackrel{4r+2m}{\dots},\Z_{4}{\setminus}2\Z_{4}).
$$
Therefore, for $\lambda\in2\Z_8$, $v_2$  satisfies property \ref{proper-1}. Now, we consider the coordinates in $v_3$. From (\ref{eq:vi}), $v_3=(x_2, u_3,\stackrel{8}{\dots},u_3)+\lambda(\one, \zero, \dots, \mathbf{7}).$ 
By Item \ref{2z22} of Lemma \ref{Lemma:cases}, we have that, for $\lambda\in 2\Z_8$, $x_2+\lambda \one$ is a permutation of $(2\Z_{8}, \stackrel{4^{t_1-1}}{\dots},  2\Z_{8})$. If $\lambda=0$, it is easy to see that $v_3$ satisfies property \ref{proper-1}.
Note that  $\lambda(\zero, \dots, \mathbf{7})$ is a permutation of $(\zero, \two, \four, \mathbf{6},\zero, \two, \four, \mathbf{6})$ if $\lambda \in \{2,6\}$, and a permutation of $(\zero, \four, \zero, \four, \zero, \four, \zero, \four)$ if $\lambda=4$. Thus, by Lemma \ref{lemm:ord8_u3}, $v_3$ satisfies property \ref{proper-1}.  Therefore, if $o(\vv)=8$, then $\vv$ satisfies property \ref{proper-1}.

Case 2: Assume that $o(\vv)=4$. We have two subcases: when $o(\zz)=4$ and $\lambda \in 2\Z_8$, and when $o(\zz)=2$ and $\lambda \in \{2,6\}$. 
For the first subcase, since $o(\zz)=4$, we have that $o(\uu)=4$. Moreover, $x_1=\zero$ and $x_2$ is a permutation of $(4\Z_8, \stackrel{2\cdot 4^{t_1-1}}{\dots}, 4\Z_8)$ by Item 2 of Lemma \ref{lem:z_t1p1_0_1}.
By induction hypothesis, $\uu$ satisfies property \ref{proper-2}. Then, $u_1=\zero$, $u_2$ contains the element in $2\Z_4{\setminus}\{0\}=\{2\}$ exactly $4^{t_1-1}$ times and $4^{t_1-1}-2^{t_1-1}$ times the element $0$,
   and  $u_3$ is a permutation of 
   \begin{equation*}\label{eq:4v2}
      (\mu,\stackrel{m}{\dots}, \mu,4\Z_8,\stackrel{n}{\dots},4\Z_8,2\Z_8{\setminus}4\Z_8, \stackrel{r}{\dots}, 2\Z_8{\setminus}4\Z_8) 
   \end{equation*}
     for some integers  $m, n,r\geq 0$ and $\mu \in \{2,6\}$.  Since $v_1=(u_1, u_1)+\lambda(\zero, \one)$, $u_1=\zero$, and $\lambda\in 2\Z_8$, we have that $v_1=\zero$.  From (\ref{eq:vi}), $v_2=(x_1, u_2,u_2,u_2,u_2)+\lambda(\one, \zero, \one, \two, \mathbf{3}).$
    If $\lambda \in \{0,4\}$, then $v_2=(x_1, u_2,u_2,u_2,u_2)$. Since $x_1=\zero$ is of length $2^{t_1}$, it is easy to see that $v_2$ in $\vv$ satisfies the same property as $u_2$ in $\uu$; that is, property \ref{proper-2}. If $\lambda \in \{2,6\}$, then 
    $v_2=(x_1, u_2, u_2,u_2,u_2)+(\two, \zero, \two, \zero, \two),$ where $x_1=\zero$ is of length $2^{t_1}$. 
    Note that $u_2+\two$ contains the element in $2\Z_4{\setminus}\{0\}=\{2\}$ as many times as $u_2$ contains the element $0$, and the element $0$ as many times as $u_2$ contains the element $2$. Thus, $v_2$ contains the element in $2\Z_4{\setminus}\{0\}=\{2\}$ exactly 
    $2^{t_1} + 2 (4^{t_1-1}) + 2( 4^{t_1-1}-2^{t_1-1})=4^{t_1}$ times
    and  
    $2 (4^{t_1-1}) + 2( 4^{t_1-1}-2^{t_1-1})=4^{t_1} -2^{t_1}$ times
    the element $0$. Therefore, for $\lambda \in 2\Z_8$, $v_2$  satisfies property \ref{proper-2}. Now, we consider the coordinates in $v_3$. From (\ref{eq:vi}), $v_3=(x_2, u_3,\stackrel{8}{\dots},u_3)+\lambda(\one, \zero, \dots, \mathbf{7})$.
    If $\lambda=0$, it is easy to see that $v_3$ satisfies property \ref{proper-2}.
    For $\lambda=4$, $x_2+\lambda \one$ is a permutation of  $(4\Z_8, \stackrel{2\cdot 4^{t_1-1}}{\dots}, 4\Z_8)$, and 
    for $\lambda \in \{2,6\}$, it is a permutation of  
    $$
    (2\Z_8{\setminus}4\Z_8, \stackrel{2\cdot 4^{t_1-1}}{\dots}, 2\Z_8{\setminus}4\Z_8).
    $$
    Note that  $\lambda(\zero, \dots, \mathbf{7})$ is a permutation of $(\zero, \two, \four, \mathbf{6},\zero, \two, \four, \mathbf{6})$ if $\lambda \in \{2,6\}$, and a permutation of $(\zero, \four, \zero, \four, \zero, \four, \zero, \four)$ if $\lambda=4$.
    Hence, by Lemma \ref{lemm:ord4_u3}, $v_3$ also satisfies property \ref{proper-2}, and so does $\vv$.
    
Now, we consider the second subcase, that is, when $o(\zz)=2$  and $\lambda \in \{2,6\}$. Since $o(\zz)=2$, we have that $o(\uu)=2$. Then, by Item 3 of Lemma \ref{lem:z_t1p1_0_1}, $x_1=\zero$ and $x_2=\zero$.  
By induction hypothesis, $\uu$ satisfies property \ref{proper-3}, so
 $u_1=\zero$,  $u_2=\zero$, and  $u_3$ contains the element in $4\Z_8{\setminus}\{0\}=\{4\}$ exactly $m=8^{t_1-1}$ times and   $m'=8^{t_1-1}-4^{t_1-1} $ times the element $0$. Since $v_1=(u_1, u_1)+\lambda(\zero, \one)$, $u_1=\zero$, and  $\lambda \in \{2,6\}$, we have that $v_1=\zero$. From (\ref{eq:vi}), $v_2=(x_1, u_2,u_2,u_2,u_2)+(\two, \zero, \two, \zero, \two)$. Since $x_1=\zero$ and $u_2=\zero$, of length $\alpha_1$ and $\alpha_2$, respectively, we have that $v_2=(\two, \zero, \two, \zero, \two).$ Therefore, $v_2$ contains the element in $2\Z_4{\setminus}\{0\}=\{2\}$ exactly $\alpha_1+2\alpha_2=4^{t_1}$ times and  $2\alpha_2= 4^{t_1}-2^{t_1}$ times  the element $0$, by (\ref{eq:t1_t2_1}). Therefore, $v_2$  satisfies property \ref{proper-2}. Now, we consider the coordinates in $v_3$. From (\ref{eq:vi}), $v_3=(x_2, u_3,\stackrel{8}{\dots},u_3)+\lambda(\one, \zero, \dots, \mathbf{7}).$ Since $x_2=\zero$,  $x_2+\lambda \one=(\lambda,\stackrel{4^{t_1}}{\dots},\lambda)$. Note that $u_3$ is a permutation of 
 $$
 (4,\stackrel{m-m'}{\dots}, 4, 4\Z_8,\stackrel{m'}{\dots}, 4\Z_8).
 $$ 
Moreover, since $\lambda \in \{2,6\}$, $\lambda(\zero,\dots, \mathbf{7})$ is a permutation of $(\zero, \two, \four, \mathbf{6}, \zero, \two, \four, \mathbf{6})$. Thus, by Item \ref{2z22} of  Lemma \ref{Lemma:cases}, $(u_3,\stackrel{8}{\dots},u_3)+(\zero,\two, \four, \mathbf{6},\zero,\two, \four, \mathbf{6})$ is a permutation of $$
 (2\Z_8,\stackrel{2(m-m')+4m'}{\dots}, 2\Z_8).
 $$
 Thus, $v_3$ is a permutation of $(\lambda,\stackrel{4^{t_1}}{\dots},\lambda, 2\Z_8,\stackrel{2(m-m')+4m'}{\dots}, 2\Z_8)$ with $\lambda \in \{2,6\}$, and hence $v_3$ also  satisfies property \ref{proper-2} and so does $\vv$. Therefore, if $o(\vv)=4$, then $\vv$ satisfies property \ref{proper-2}.

Case 3: Assume that $o(\vv)=2$. Then,  $o(\zz)=2$ and $\lambda \in \{0,4\}$. 
Since $o(\zz)=2$, then $o(\uu)=2$.  Moreover, $x_1= \zero$ and $x_2=\zero$ by Item 3 of Lemma \ref{lem:z_t1p1_0_1}. By induction hypothesis, $\uu$ satisfies property \ref{proper-3}, and then $u_1=\zero$, $u_2=\zero$, and $u_3$ contains the element in $4\Z_8{\setminus}\{0\}=\{4\}$ exactly $8^{t_1-1}$ times and  $8^{t_1-1}-4^{t_1-1} $ times the element  $0$. Since $v_1=(u_1, u_1)+\lambda(\zero, \one)$, $v_1=\zero$. From (\ref{eq:vi}), $v_2=(x_1, u_2,u_2,u_2,u_2)+\lambda(\one, \zero, \one, \two, \mathbf{3})$, where $x_1=\zero$ and $u_2=\zero$, so $v_2=\zero$.
  From (\ref{eq:vi}), $v_3=(x_2, u_3,\stackrel{8}{\dots},u_3)+\lambda(\one,\zero, \dots, \mathbf{7})$, where $x_2=\zero$ is of length $4^{t_1}$. If $\lambda=0$, it is easy to see that $v_3$ satisfies property \ref{proper-3}. If $\lambda=4$, $v_3=(x_2, u_3,\stackrel{8}{\dots},u_3)+(\four,\zero, \four, \zero, \four,\zero, \four,\zero, \four)$. Note that $u_3+\four$ contains the element in $4\Z_8{\setminus}\{0\}=\{4\}$ as many times as $u_3$ contains the element $0$, and the element $0$ as many times as $u_3$ contains the element $4$.
Then, $v_3$ contains the element $4$ exactly 
$4^{t_1}+4( 8^{t_1-1}) + 4(8^{t_1-1} - 4^{t_1-1})=8^{t_1} $ times and 
$4( 8^{t_1-1}) + 4(8^{t_1-1} - 4^{t_1-1})=8^{t_1}-4^{t_1} $
the element 0. Therefore, $\vv$ satisfies property \ref{proper-3}. This completes the proof.
\end{proof}

\begin{lemma}\label{lemma:same_num2}
Let $t_1\geq 1$ and $t_2\geq 0$ be integers. Let ${\cH}^{t_1,t_2,1}$ be the $\Z_2\Z_4 \Z_8$-additive code of type $(\alpha_1,\alpha_2,\alpha_3; t_1,t_2,1)$ generated by the matrix $A^{t_1,t_2,1}$. 
Let $\uu=(u_1\mid u_2\mid u_3) \in {\cal{G}}^{t_1,t_2,1}$.

 \begin{enumerate}
 \item \label{proper_1} If $o(\uu)=8$, then $\uu$ has the following property:
 \begin{enumerate}
     \item \label{proper_1a}$u_1$ contains every element of $\Z_2$ the same number of times, $u_2$ is a permutation of $(\mu, \stackrel{m}{\dots}, \mu, 2\Z_{4}, \stackrel{n}{\dots}, 2\Z_{4}, \Z_{4}{\setminus}2\Z_{4},\stackrel{r}{\dots},\Z_{4}{\setminus}2\Z_{4})$ for some integers  $m,n,r\geq 0$ and $\mu \in \{1,3\}$, and $u_3$ is a permutation of $(\mu', \stackrel{m'}{\dots}, \mu', \mu'', \stackrel{m'}{\dots}, \mu'', 2\Z_{8}, \stackrel{n'}{\dots}, 2\Z_{8}, \Z_{8}{\setminus}2\Z_{8},\stackrel{r'}{\dots},\Z_{8}{\setminus}2\Z_{8})$ for some integers  $m', n',r'\geq 0$ and $\mu, \mu'\in \{1,3,5,7\}$.
 \end{enumerate}
  \item \label{proper_2} If $o(\uu)=4$, then $\uu$ has one of the following properties:
  \begin{enumerate}
      \item \label{proper_2a}$u_1=\zero$, $u_2$ contains the element in $2\Z_4{\setminus}\{0\}=\{2\}$ exactly $\frac{1}{2}(\frac{\alpha_1}{2}+\alpha_2)=4^{t_1+t_2-1}$ times and  $\frac{\alpha_2}{2}-\frac{\alpha_1}{4}=4^{t_1+t_2-1}-2^{t_1+t_2-1}$ times the element $0$,
   and  $u_3$ is a permutation of $(\mu,\stackrel{m}{\dots}, \mu,4\Z_8,\stackrel{n}{\dots},4\Z_8,2\Z_8{\setminus}4\Z_8, \stackrel{r}{\dots}, 2\Z_8{\setminus}4\Z_8)$ for some integers  $m, n,r\geq 0$ and $\mu \in \{2,6\}$.
    \item \label{proper_2b} $u_1$ contains every element of $\Z_2$ the same number of times,  $u_2$ is a permutation of  $(\mu,\stackrel{m}{\dots},\mu,2\Z_4,\stackrel{n}{\dots},2\Z_4,\Z_4{\setminus}2\Z_4, \stackrel{r}{\dots}, \Z_4{\setminus}2\Z_4)$ for some integers  $m, n,r\geq 0$ and $\mu \in \{1,3\}$, 
    and 
    $u_3$ is a permutation of $(4\Z_{8}, \stackrel{t}{\dots}, 4\Z_{8}, 2\Z_{8}{\setminus}4\Z_{8},\stackrel{t'}{\dots},2\Z_{8}{\setminus}4\Z_{8})$ for some integers $t, t'\geq 0$.
    
  \end{enumerate}
  \item \label{proper_3} If $o(\uu)=2$, then $\uu$ has one of the following properties:
  \begin{enumerate}
   \item \label{proper_3a}$u_1=\zero$, $u_2=\zero$, and $u_3$ contains the element in $4\Z_8{\setminus}\{0\}=\{4\}$ exactly $\frac{1}{4}(\frac{\alpha_1}{2}+\alpha_2+2\alpha_3)=8^{t_1-1}4^{t_2}$ times and  $\frac{\alpha_3}{2}-\frac{1}{4}(\frac{\alpha_1}{2}+\alpha_2)=8^{t_1-1}4^{t_2}-4^{t_1+t_2-1}$ times the element  $0$.
  \item \label{proper_3b}$u_1=\zero$,  $u_2$ contains the element in $2\Z_4{\setminus}\{0\}=\{2\}$ exactly $\frac{1}{2}(\frac{\alpha_1}{2}+\alpha_2)=4^{t_1+t_2-1}$ times and  $\frac{\alpha_2}{2}-\frac{\alpha_1}{4}=4^{t_1+t_2-1}-2^{t_1+t_2-1}$ times the element  $0$, and $u_3$ is a permutation of $(4\Z_8,\stackrel{m}{\dots},4\Z_8)$ for some $m\geq 0$.
  \end{enumerate}
   \end{enumerate}
\end{lemma}
\begin{proof}
We prove this lemma by induction on $t_2\geq 0$. The lemma is true for the code  ${\cH}^{t_1,0,1}$ by Lemma \ref{lemma:same_num1}.
Assume that the lemma holds for the code ${\cH}^{t_1,t_2,1}$ of type $(\alpha_1, \alpha_2, \alpha_3; t_1,t_2,1)$ with $t_1\geq 1$ and $t_2\geq 0$.
By Lemma \ref{relation:t1a1}, we have that
 \begin{align}\label{eq:t1t2}
      2^{t_1+t_2}=\alpha_1, 
    4^{t_1+t_2}=\alpha_1+2\alpha_2,~\mbox{and}~ 
    8^{t_1}4^{t_2}=\alpha_1+2\alpha_2+4\alpha_3.
\end{align}
 Now, we have to show that the lemma is also true for the code ${\cH}^{t_1,t_2+1,1}$. 
 
  Let $\vv=(v_1\mid v_2 \mid v_3)\in {\cal{G}}^{t_1,t_2+1,1}$. We can write 
 \begin{align*}
     \vv=\zz+\lambda \ww,
 \end{align*}
 where  $\zz=(u_1,u_1 \mid x_1,u_2, u_2, u_2, u_2\mid u_3,u_3, u_3, u_3)$, $\ww=(\zero, \one \mid \one, \zero, \one, \two, \mathbf{3}\mid  \zero, \two, \four, \mathbf{6})$, $\uu=(u_1\mid u_2\mid u_3) \in {\cal{G}}^{t_1,t_2,1}$, $\lambda \in \{0,1,2,3\}$, and $x_1\in (2\Z_4)^{2^{t_1+t_2}}$ such that either $x_1=\zero$ or  a permutation  of $(2\Z_4,\stackrel{2^{t_1+t_2-1}}{\dots}, 2\Z_4)$. Then,  
\begin{equation}\label{eq:vi_}
\begin{split}
   v_1&=(u_1, u_1)+\lambda(\zero, \one),\\
     v_2&=(x_1,u_2,u_2,u_2,u_2)+\lambda (\one, \zero, \one, \two, \mathbf{3}),\\
     v_3&=(u_3,u_3,u_3,u_3)+\lambda(\zero, \two, \four, \mathbf{6}). 
\end{split}
\end{equation}
If $\mathbf{z}=\zero$, then $\vv=\lambda \ww$. It is easy to see that $\vv$ satisfies property \ref{proper_2b} if $\lambda\in \{1,3\}$ and  property \ref{proper_3b} if $\lambda=2$.
Therefore, we focus on the case  when $\mathbf{z}\neq \zero$. 

Case 1: Assume that $o(\vv)=8$. Then, $o(\zz)=8$ and $\lambda \in \{0,1,2,3\}$. We have that $o(\uu)=8$ and,  by Item 1 of Lemma \ref{lem:z_t1t2p1_1}, $x_1$ is a permutation of $(2\Z_4,\stackrel{2^{t_1+t_2-1}}{\dots}, 2\Z_4)$. By induction hypothesis, $\uu$ satisfies property \ref{proper_1a}. Then, $u_1$ contains every element of $\Z_2$ the same number of times, $u_2$ is a permutation of 
\begin{equation*}\label{eq:v2_}
(\mu, \stackrel{m}{\dots}, \mu, 2\Z_{4}, \stackrel{n}{\dots}, 2\Z_{4}, \Z_{4}{\setminus}2\Z_{4},\stackrel{r}{\dots},\Z_{4}{\setminus}2\Z_{4}),
\end{equation*} where  $m,n,r\geq 0$ and $\mu \in \{1,3\}$, and $u_3$ is a permutation of 
\begin{equation*}\label{eq:v3_}
(\mu', \stackrel{m'}{\dots}, \mu', \mu'', \stackrel{m'}{\dots}, \mu'', 2\Z_{8}, \stackrel{n'}{\dots}, 2\Z_{8}, \Z_{8}{\setminus}2\Z_{8},\stackrel{r'}{\dots},\Z_{8}{\setminus}2\Z_{8}),
\end{equation*}
where  $m', n',r'\geq 0$ and $\mu', \mu''\in \{1,3,5,7\}$. First, since $ v_1=(u_1, u_1)+\lambda(\zero, \one)$, $v_1$ contains every element of $\Z_2$ the same number of times, for any $\lambda \in \{0,1,2,3\}$. Second, from (\ref{eq:vi_}), $v_2=(x_1, u_2,u_2,u_2,u_2)+\lambda(\one,\zero,\one,\two,\mathbf{3}).$ If $\lambda=0$, then $v_2$ clearly satisfies \ref{proper_1a}. 
If $\lambda \in \{1,3\}$, then  we have that $(u_2,u_2,u_2,u_2)+\lambda(\zero,\one,\two,\mathbf{3})$ is a permutation of $(\Z_4,\stackrel{\alpha_2}{\dots}, \Z_4)$ by Lemma \ref{lemm:uu+lambda}. For $\lambda \in \{1,3\}$, since $x_1+\lambda \one$ is a permutation of $(\Z_4{\setminus}2\Z_4, \stackrel{2^{t_1+t_2-1}}{\dots},\Z_4{\setminus}2\Z_4)$ by Item  \ref{2z2_} of Lemma \ref{Lemma:cases}, we have that 
$v_2$  satisfies property \ref{proper_1a}. 
If $\lambda=2$, $v_2=(x_1, u_2,u_2,u_2,u_2)+(\two,\zero,\two,\zero,\two).$ By Item  \ref{2z22} of Lemma \ref{Lemma:cases}, we have that $x_1+\two$ is a permutation of $(2\Z_{4}, \stackrel{2^{t_1+t_2-1}}{\dots},2\Z_{4})$. Therefore, by Lemma \ref{lemm:ord8_u2}, $v_2$ is a permutation of $(2\Z_{4}, \stackrel{4n+2^{t_1+t_2-1}}{\dots}, 2\Z_{4}, \Z_{4}{\setminus}2\Z_{4},\stackrel{4r+2m}{\dots},\Z_{4}{\setminus}2\Z_{4})$ and then $v_2$  satisfies property \ref{proper_1a}. Finally, we consider the coordinates in $v_3$. From (\ref{eq:vi_}), $v_3=(u_3,u_3, u_3, u_3)+\lambda(\zero, \two, \four, \mathbf{6}).$ If $\lambda=0$, then $v_3$ clearly satisfies \ref{proper_1a}. Note that $\lambda(\zero, \two, \four, \mathbf{6})=(\zero, \four, \zero, \four)$ if $\lambda=2$ and $\lambda(\zero, \two, \four, \mathbf{6})$ is a permutation of $(\zero, \two, \four, \mathbf{6})$ if $\lambda \in \{1,3\}$. Therefore, by Lemma \ref{lemm:ord8_u3}, $v_3$  satisfies property  \ref{proper_1a}, and so does $\vv$.

Case 2: Assume that $o(\vv)=4$. We have two subcases: when $o(\zz)=4$ and $\lambda \in \{0,1,2,3\}$, and when $o(\zz)=2$ and $\lambda \in \{1,3\}$. For the first subcase, since $o(\zz)=4$, $o(\uu)=4$. By induction hypothesis, $\uu$ satisfies property \ref{proper_2a} or \ref{proper_2b}. Assume that $\uu$ satisfies property \ref{proper_2a}. Then, $u_1=\zero$, $u_2$ contains the element in $2\Z_4{\setminus}\{0\}=\{2\}$ exactly $ 4^{t_1+t_2-1}$ times and $ 4^{t_1+t_2-1} - 2^{t_1+t_2-1}$ times the element $0$,
   and  $u_3$ is a permutation of 
   \begin{equation*}\label{eq:4v2_}
      (\mu,\stackrel{m}{\dots}, \mu,4\Z_8,\stackrel{n}{\dots},4\Z_8,2\Z_8{\setminus}4\Z_8, \stackrel{r}{\dots}, 2\Z_8{\setminus}4\Z_8) 
   \end{equation*}
     for some integers  $m, n,r\geq 0$ and $\mu \in \{2,6\}$. Note that, in this case,  $x_1=\zero$ by Item 2 of Lemma \ref{lem:z_t1t2p1_1}. 
     If $\lambda=0$, then it is easy to see that $\vv$ satisfies property \ref{proper_2a}. If $\lambda=2$, we show that $\vv$ satisfies property \ref{proper_2a}.  Since $v_1=(u_1, u_1)+\lambda(\zero, \one)$, $u_1=\zero$, and $\lambda=2$, we have that $v_1=\zero$. 
     From (\ref{eq:vi_}), $v_2=(x_1, u_2,u_2,u_2,u_2)+(\two, \zero, \two, \zero, \two)$, where $x_1=\zero$ is of length $2^{t_1+t_2}$. 
     Note that $u_2+\two$ contains the element in $2\Z_4{\setminus}\{0\}=\{2\}$ as many times as $u_2$ contains the element $0$, and the element $0$ as many times as $u_2$ contains the element $2$. Thus, $v_2$ contains the element in $2\Z_4{\setminus}\{0\}=\{2\}$ exactly 
     $2^{t_1+t_2} + 2(4^{t_1+t_2-1}) + 2(4^{t_1+t_2-1} - 2^{t_1+t_2-1} ) =4^{t_1+t_2}$ times
     and  $2(4^{t_1+t_2-1}) + 2(4^{t_1+t_2-1} - 2^{t_1+t_2-1} ) =4^{t_1+t_2}-2^{t_1+t_2}$ times
     the element $0$, 
     so $v_2$ satisfies property \ref{proper_2a}. 
     From (\ref{eq:vi_}), $v_3=( u_3,u_3,u_3,u_3)+( \zero, \four, \zero, \four)$. By Item \ref{lemm:ord4_case2} of Lemma \ref{lemm:ord4_u3}, $v_3$ is a permutation of 
     $$
     (4\Z_8,\stackrel{4n}{\dots},4\Z_8,2\Z_8{\setminus}4\Z_8, \stackrel{4r+2m}{\dots}, 2\Z_8{\setminus}4\Z_8).
     $$
     Therefore, for $\lambda =2$, $\vv$ satisfies property \ref{proper_2a}. Finally, if $\lambda \in \{1,3\}$, we show that $\vv$ satisfies property \ref{proper_2b}. 
     Since $v_1=(u_1, u_1)+\lambda(\zero, \one)$, $u_1=\zero$, and $\lambda \in \{1,3\}$, we have that $v_1$ contains every element of $\Z_2$ the same number of times. 
     From (\ref{eq:vi_}), $v_2=(x_1, u_2,u_2,u_2,u_2)+\lambda(\one, \zero, \one, \two, \mathbf{3})$, where $x_1=\zero$ is of length $2^{t_1+t_2}$. Since $\lambda \in \{1,3\}$, by Lemma \ref{lemm:uu+lambda}, we have that $v_2$ is a permutation of $(\lambda, \stackrel{2^{t_1+t_2}}{\dots},\lambda,\Z_4,\stackrel{\alpha_2}{\dots}, \Z_4)$.
     From (\ref{eq:vi_}), $v_3=(u_3,u_3, u_3, u_3)+\lambda(\zero, \two, \four, \mathbf{6}).$ Note that, for $\lambda \in \{1,3\}$,  $\lambda(\zero, \two, \four, \mathbf{6})$ is a permutation of $(\zero, \two, \four, \mathbf{6})$. Thus, by Item \ref{lemm:ord4_case1} of Lemma \ref{lemm:ord4_u3}, $v_3$  satisfies property  \ref{proper_2b}, and so does $\vv$. Therefore, if $o(\uu)=4$ and $\uu$ satisfies property 
     \ref{proper_2a}, we have that $\vv$ satisfies either property \ref{proper_2a} or \ref{proper_2b}. 
     
     We continue with the first subcase, when $o(\zz)=4$ and $\lambda \in \{0,1,2,3\}$. Again, we have that  $o(\uu)=4$. Now, we assume that $\uu$ satisfies property 
     \ref{proper_2b}. 
     Then, $u_1$ contains every element of $\Z_2$ the same number of times,  $u_2$ is a permutation of 
     $$
     (\mu,\stackrel{m}{\dots},\mu,2\Z_4,\stackrel{n}{\dots},2\Z_4,\Z_4{\setminus}2\Z_4, \stackrel{r}{\dots}, \Z_4{\setminus}2\Z_4)
     $$ 
    for some integers  $m, n,r\geq 0$ and $\mu \in \{1,3\}$,
    and $u_3$ is a permutation of 
    $(4\Z_{8}, \stackrel{t}{\dots}, 4\Z_{8}, 2\Z_{8}{\setminus}4\Z_{8},\stackrel{t'}{\dots},2\Z_{8}{\setminus}4\Z_{8})$
    for some integers $t, t'\geq 0$. 
   Note that, in this case,  $x_1$ is a permutation of $(2\Z_4, \stackrel{2^{t_1+t_2-1}}{\dots}, 2\Z_4)$ by Item 2 of Lemma \ref{lem:z_t1t2p1_1}.
    Now, we show that $\vv$ satisfies property \ref{proper_2b}.  
    Since $v_1=(u_1, u_1)+\lambda(\zero, \one)$ and $u_1$ contains every element of $\Z_2$ the same number of times, we have that $v_1$ contains every element of $\Z_2$ the same number of times, for any $\lambda \in \{0,1,2,3\}$. 
    From (\ref{eq:vi_}), $v_2=(x_1, u_2,u_2,u_2,u_2)+\lambda(\one, \zero, \one, \two, \mathbf{3})$. 
If $\lambda=0$, it is clear that $v_2$ satisfies property \ref{proper_2b}. Note that $x_1+\lambda \one$ is a permutation of $(2\Z_4, \stackrel{2^{t_1+t_2-1}}{\dots}, 2\Z_4)$ if $\lambda=2$, and  a permutation of $(\Z_4{\setminus}2\Z_4, \stackrel{2^{t_1+t_2-1}}{\dots}, \Z_4{\setminus}2\Z_4)$ if $\lambda \in \{1,3\}$.   If $\lambda=2$, then by Lemma \ref{lemm:ord8_u2}, $(u_2,u_2,u_2,u_2)+(\zero, \two, \zero, \two)$ is a permutation of 
    $$
    (2\Z_{4}, \stackrel{4n}{\dots}, 2\Z_{4}, \Z_{4}{\setminus}2\Z_{4},\stackrel{4r+2m}{\dots},\Z_{4}{\setminus}2\Z_{4}).
    $$
    If $\lambda \in \{1,3\}$, then by Lemma \ref{lemm:uu+lambda}, $(u_2,u_2,u_2,u_2)+\lambda(\zero, \one, \two, \mathbf{3})$ is a permutation of $(\Z_4,\stackrel{\alpha_2}{\dots}, \Z_4)$.
    Therefore, $v_2$ satisfies property \ref{proper_2b}. From (\ref{eq:vi_}), $v_3=(u_3,u_3, u_3, u_3)+\lambda(\zero, \two, \four, \mathbf{6}).$ If $\lambda=0$, it is clear that $v_3$ satisfies property \ref{proper_2b}. Note that   $\lambda(\zero, \two, \four, \mathbf{6})$ is a permutation of $(\zero, \two, \four, \mathbf{6})$ if $\lambda \in \{1,3\}$, and $\lambda(\zero, \two, \four, \mathbf{6})=(\zero, \four, \zero, \four)$ if $\lambda=2$. Therefore, by Lemma \ref{lemm:ord4_u3}, $v_3$  satisfies property \ref{proper_2b}, and so does $\vv$.
    
    Now, we consider the second subcase when $o(\vv)=4$, that is, when $o(\zz)=2$ and $\lambda\in \{1,3\}$. Since $o(\zz)=2$, $o(\uu)=2$. By induction hypothesis, $\uu$ satisfies property \ref{proper_3a} or \ref{proper_3b}. Assume that $\uu$ satisfies property \ref{proper_3a}. Then, $u_1=\zero$,  $u_2=\zero$, and  $u_3$ contains the element in $4\Z_8{\setminus}\{0\}=\{4\}$ exactly 
    $m=8^{t_1-1}4^{t_2}$
    times and   
    $m'=8^{t_1-1}4^{t_2}-4^{t_1+t_2-1}$
    times the element $0$.
     By Item 3 of Lemma \ref{lem:z_t1t2p1_1}, we have that  $x_1=\zero$. 
      Since $v_1=(u_1, u_1)+\lambda(\zero, \one)$, $u_1=\zero$, and $\lambda \in \{1,3\}$, we have that $v_1$ contains every element of $\Z_2$ the same number of times.  From (\ref{eq:vi_}), $v_2=(x_1, u_2,u_2,u_2,u_2)+\lambda(\one, \zero, \one, \two, \mathbf{3})$, where $x_1=\zero$ is of length $2^{t_1+t_2}$. 
    By Lemma \ref{lemm:uu+lambda}, we have that $v_2$ is a permutation of 
    $$
    (\lambda, \stackrel{2^{t_1+t_2}}{\dots},\lambda,\Z_4,\stackrel{\alpha_2}{\dots}, \Z_4),
    $$ where $\lambda \in \{1,3\}$.
    From (\ref{eq:vi_}), $v_3=(u_3,u_3, u_3, u_3)+\lambda(\zero, \two, \four, \mathbf{6}).$ Note that $u_3$ is a permutation of $(4,\stackrel{m-m'}{\dots}, 4, 4\Z_8,\stackrel{m'}{\dots}, 4\Z_8)$ and, since $\lambda \in \{1,3\}$,  $\lambda(\zero, \two, \four, \mathbf{6})$ is a permutation of $(\zero, \two, \four, \mathbf{6})$. Thus, by Item \ref{2z22} of Lemma \ref{Lemma:cases}, $v_3=(u_3,u_3, u_3, u_3)+(\zero,\two, \four, \mathbf{6})$ is a permutation of $(2\Z_8,\stackrel{m+m'}{\dots}, 2\Z_8)$, so  $v_3$ satisfies property \ref{proper_2b}, and so does $\vv$. Therefore, if $o(\uu)=2$ and $\uu$ satisfies property 
     \ref{proper_3a}, we have that $\vv$ satisfies property \ref{proper_2b}. 
     
 We continue with the second subcase, when $o(\zz)=2$ and $\lambda \in \{1,3\}$. Again, we have that  $o(\uu)=2$. Now, we assume that $\uu$ satisfies property     
     \ref{proper_3b}. 
    Then, $u_1=\zero$,  $u_2$ contains the element in $2\Z_4{\setminus}\{0\}=\{2\}$ exactly $4^{t_1+t_2-1}$ 
    times and  $4^{t_1+t_2-1}-2^{t_1+t_2-1}$ 
    times the element  $0$, and $u_3$ is a permutation of $(4\Z_8,\stackrel{m}{\dots},4\Z_8)$ for some $m\geq 0$. By Item 3 of Lemma \ref{lem:z_t1t2p1_1}, we have that  $x_1=\zero$. 
    Since $v_1=(u_1, u_1)+\lambda(\zero, \one)$, $u_1=\zero$, and $\lambda \in \{1,3\}$, we have that $v_1$ contains every element of $\Z_2$ the same number of times. 
    From (\ref{eq:vi_}), $v_2=(x_1, u_2,u_2,u_2,u_2)+\lambda(\one, \zero, \one, \two, \mathbf{3})$, where $x_1=\zero$ is of length $2^{t_1+t_2}$. By Lemma \ref{lemm:uu+lambda}, we have that $v_2$ is a permutation of $$(\lambda, \stackrel{2^{t_1+t_2}}{\dots},\lambda,\Z_4,\stackrel{\alpha_2}{\dots}, \Z_4),$$ where $\lambda \in \{1,3\}$.
    From (\ref{eq:vi_}), $v_3=(u_3,u_3, u_3, u_3)+\lambda(\zero, \two, \four, \mathbf{6}).$ Since $\lambda \in \{1,3\}$,  $\lambda(\zero, \two, \four, \mathbf{6})$ is a permutation of $(\zero, \two, \four, \mathbf{6})$. 
    Thus, by Item \ref{2z22} of Lemma \ref{Lemma:cases}, $v_3=(u_3,u_3, u_3, u_3)+(\zero,\two, \four, \mathbf{6})$ is a permutation of $(2\Z_8,\stackrel{2m}{\dots}, 2\Z_8)$.  Therefore,  $v_3$  satisfies property \ref{proper_2b}, and so does $\vv$. 
    
Case 3: Assume that $o(\vv)=2$. Then,  $o(\zz)=2$ and $\lambda \in \{0,2\}$.  Since $o(\zz)=2$, we have that  $o(\uu)=2$ and, by Item 3 of Lemma \ref{lem:z_t1t2p1_1}, $x_1=\zero$. By induction hypothesis, $\uu$ satisfies property \ref{proper_3a} or \ref{proper_3b}. Assume that $\uu$ satisfies property \ref{proper_3a}. Then, $u_1=\zero$,  $u_2=\zero$, and  $u_3$ contains the element in $4\Z_8{\setminus}\{0\}=\{4\}$ exactly 
$m=8^{t_1-1}4^{t_2}$
times and   
$m'=8^{t_1-1}4^{t_2}-4^{t_1+t_2-1}$
times the element $0$.
If $\lambda=0$, then $\vv=( \zero \mid \zero \mid v_3)$ satisfies property \ref{proper_3a}, since $v_3$ contains $4m$ times the element $4$ and $4m'$ the element $0$. 
Now, we assume that $\lambda=2$. Since $v_1=(u_1, u_1)+\lambda(\zero, \one)$, $u_1=\zero$, and $\lambda=2$, we have that $v_1=\zero$. 
From (\ref{eq:vi_}),  $v_2=(x_1, u_2,u_2,u_2,u_2)+(\two, \zero, \two, \zero, \two)$, where $x_1=\zero$ is of length $2^{t_1+t_2}$ and $u_2=\zero$. Therefore, $v_2$ contains the element in $2\Z_4{\setminus}\{0\}=\{2\}$ exactly $\alpha_1+2\alpha_2=4^{t_1+t_2}$ times  and  $2\alpha_2=4^{t_1+t_2}-2^{t_1+t_2} $ times the element $0$, by (\ref{eq:t1t2}). From (\ref{eq:vi_}), $v_3=(u_3,u_3, u_3, u_3)+(\zero, \four, \zero, \four).$ Note that $u_3$ is a permutation of 
$$
(4,\stackrel{m-m'}{\dots}, 4, 4\Z_8,\stackrel{m'}{\dots}, 4\Z_8).
$$ 
Thus, by  Item \ref{2z22} of Lemma \ref{Lemma:cases}, $v_3$ is a permutation of $(4\Z_8,\stackrel{2m+2m'}{\dots}, 4\Z_8)$, so $v_3$ satisfies property \ref{proper_3b}, and so does $\vv$.   Therefore, if $o(\uu)=2$ and $\uu$ satisfies property \ref{proper_3a}, we have that $\vv$ satisfies property \ref{proper_3b}. 

We continue with the case when $o(\zz)=2$ and $\lambda \in \{0,2\}$. Again, we have that $o(\uu)=2$ and $x_1=\zero$. Now, we assume that $\uu$ satisfies property \ref{proper_3b}. 
Then, $u_1=\zero$,  $u_2$ contains the element in $2\Z_4{\setminus}\{0\}=\{2\}$ exactly $4^{t_1+t_2-1}$ times and  $4^{t_1+t_2-1}-2^{t_1+t_2-1}$ times the element  $0$, and $u_3$ is a permutation of $(4\Z_8,\stackrel{m}{\dots},4\Z_8)$ for some $m\geq 0$. 
If $\lambda=0$, then it is easy to see that $\vv$ satisfies property \ref{proper_3b}. Now, we assume that $\lambda=2$.
Since $v_1=(u_1, u_1)+\lambda(\zero, \one)$, $u_1=\zero$, and $\lambda=2$, we have that $v_1=\zero$.  From (\ref{eq:vi_}), $v_2=(x_1, u_2,u_2,u_2,u_2)+(\two, \zero, \two, \zero, \two)$, where $x_1=\zero$ is of length $2^{t_1+t_2}$. Note that $u_2+\two$ contains the element in $2\Z_4{\setminus}\{0\}=\{2\}$ as many times as $u_2$ contains the element $0$, and the element $0$ as many times as $u_2$ contains the element $2$. Therefore, $v_2$ contains the element in $2\Z_4{\setminus}\{0\}=\{2\}$ exactly
 $2^{t_1+t_2} + 2(4^{t_1+t_2-1}) + 2(4^{t_1+t_2-1} - 2^{t_1+t_2-1} ) =4^{t_1+t_2}$ times
     and  $2(4^{t_1+t_2-1}) + 2(4^{t_1+t_2-1} - 2^{t_1+t_2-1} ) =4^{t_1+t_2}-2^{t_1+t_2}$ times
     the element $0$.
From (\ref{eq:vi_}), $v_3=(u_3,u_3, u_3, u_3)+(\zero, \four, \zero, \four).$ By Item \ref{2z22} of Lemma \ref{Lemma:cases}, 
$v_3$ is a permutation of $(4\Z_8,\stackrel{4m}{\dots}, 4\Z_8)$.  Therefore,  $v_3$ satisfies property \ref{proper_3b}, and so does $\vv$. This completes the proof.
\end{proof}
 


 \begin{proposition}\label{Th:H1}
Let $t_1\geq 1$ and $t_2\geq 0$ be integers. The $\Z_2\Z_4 \Z_8$-additive code 
${\cH}^{t_1,t_2,1}$, generated by the matrix $A^{t_1,t_2,1}$, is a $\Z_2\Z_4 \Z_8$-additive Hadamard code. 
\end{proposition}
\begin{proof}
Let ${\cH}^{t_1,t_2,1}$ be the $\Z_2\Z_4 \Z_8$-additive code of type $(\alpha_1, \alpha_2, \alpha_3;t_1, t_2,1)$ and  $H^{t_1,t_2,1}$ be the corresponding  $\Z_2\Z_4 \Z_8$-linear code of length $N$. We have that $N=\alpha_1+2\alpha_2+4\alpha_3$. The cardinality of $H^{t_1, t_2,1}$ is $8^{t_1}\cdot 4^{t_2}\cdot 2=2(\alpha_1+2\alpha_2+4\alpha_3)=2N$ by Lemma \ref{relation:t1a1}. By Proposition \ref{disweight}, the minimum distance of $H^{t_1, t_2,1}$ is equal to the minimum weight of $H^{t_1, t_2,1}$. Therefore, we just need to prove that the minimum weight of $H^{t_1, t_2,1}$ is $N/2$. 

We can write that ${\cH}^{t_1,t_2,1}={\cG}^{t_1,t_2,1} \cup ({\cG}^{t_1,t_2,1}+(\one \mid \two \mid \four))$. By Corollary \ref{lemma2}, $H^{t_1,t_2,1}=\Phi({\cG}^{t_1,t_2,1})\cup (\Phi({\cG}^{t_1,t_2,1})+\one)$. Let $\uu=(u_1 \mid u_2 \mid u_3) \in {\cH}^{t_1,t_2,1}{\setminus} \{\zero, (\one\mid \two\mid \four)\}$. We show that $\wt_H(\Phi(\uu))=N/2$. First, consider $\uu \in {\cG}^{t_1,t_2,1}{\setminus}\{\zero\}$. 
If $o(\uu)=8$, then by Lemma \ref{lemma:same_num2}, $u_1$ contains every element of $\Z_2$ the same number of times, and for $i \in \{2,3\}$, $u_i$ contains every element of $2^{i-1}\Z_{2^i}$ exactly $s_i$ times, $s_i\geq 0$, and the remaining $\alpha_i - 2s_i$ coordinates of $u_i$ are from $\Z_{2^i}{\setminus}2^{i-1}\Z_{2^i}$. Thus, from the definition of $\Phi$, we have that $\wt_H(\Phi(\uu))=\alpha_1/2+2s_2+(\alpha_2-2s_2)\cdot 1+ 4s_3+(\alpha_3-2s_3)\cdot 2=\alpha_1/2+\alpha_2+2\alpha_3=N/2$. If $o(\uu)=4$, then $\uu$ satisfies property \ref{proper_2a} 
or \ref{proper_2b} given in Lemma \ref{lemma:same_num2}. If $\uu$ satisfies property \ref{proper_2a}, then $u_3$ contains every element of $4\Z_8$ exactly $m$ times, $m\geq 0$, and the remaining coordinates of $u_3$ are from $\Z_8{\setminus}4\Z_8$. Thus, $\wt_H(\Phi(\uu))=\alpha_1/2+\alpha_2+4m+(\alpha_3-2m)\cdot 2=\alpha_1/2+\alpha_2+2\alpha_3=N/2$. Otherwise, if $\uu$ satisfies property \ref{proper_2b}, then $\wt_H(\Phi(\uu))=\alpha_1/2+2n+(\alpha_2-2n)\cdot 1+4t+(\alpha_3 -2t)\cdot 2=N/2$. If $o(\uu)=2$, then $\uu$ satisfies property \ref{proper_3a}  or \ref{proper_3b}  given in Lemma \ref{lemma:same_num2}. If $\uu$ satisfies property \ref{proper_3a}, then $\wt_H(\Phi(\uu))=\frac{1}{4}(\alpha_1/2+\alpha_2+2\alpha_3)\cdot4=N/2$. Otherwise, if $\uu$ satisfies property \ref{proper_3b}, then ${\setminus}=2\cdot \frac{1}{2}(\alpha_1/2+\alpha_2)+4m+(\alpha_3-2m)\cdot 2=N/2$.

Finally, note that $\wt_H(\Phi(\uu)+\one)=N/2$. Therefore, we have that the weight of every element of $H^{t_1,t_2,1}{\setminus}\{\zero, \one\}$ is $N/2$, that is, the minimum weight of $H^{t_1,t_2,1}$ is $N/2$.
\end{proof}

\begin{proposition}\label{Th:H2}
Let $t_1\geq 1$, $t_2\geq 0$, and $t_3\geq1$ be integers. If ${\cH}^{t_1,t_2,t_3}$ is a $\Z_2\Z_4\Z_8$-additive Hadamard code of type $(\alpha_1, \alpha_2,\alpha_3; t_1,t_2,t_3)$, then, by applying  construction (\ref{eq:recGenMatrix3}), ${\cH}^{t_1,t_2,t_3+1}$ is a $\Z_2\Z_4\Z_8$-additive Hadamard code of type $(2\alpha_1,2\alpha_2,2\alpha_3; t_1,t_2,t_3+1)$.
\end{proposition}
\begin{proof}
By construction (\ref{eq:recGenMatrix3}), ${\cH}^{t_1,t_2,t_3+1}$ is a $\Z_2\Z_4\Z_8$-additive code of type $(\alpha'_1, \alpha'_2, \alpha'_3; t_1, t_2,t_3+1)$, where $\alpha'_1=2\alpha_1$, $\alpha'_2=2\alpha_2$, and  $\alpha'_3=2\alpha_3$. 

Since $H^{t_1,t_2,t_3}$ is a Hadamard code of length $N=\alpha_1+2\alpha_2+4\alpha_3$, then its minimum distance is $N/2$ and $|H^{t_1,t_2,t_3}|=2N$. Note that $H^{t_1,t_2,t_3+1}$ is a $\Z_2\Z_4\Z_8$-linear code of length $N'=\alpha'_1+2\alpha'_2+4\alpha'_3=2N$ and $|H^{t_1,t_2,t_3+1}|=8^{t_1}4^{t_2}2^{t_3+1}=2|H^{t_1,t_2,t_3}|=2\cdot 2N=2N'$. By Proposition \ref{disweight}, the minimum distance of $H^{t_1, t_2,t_3+1}$ is equal to the minimum weight of $H^{t_1, t_2,t_3+1}$. Now, we only have to prove that the minimum weight of $H^{t_1,t_2,t_3+1}$ is $N'/2$. Let ${\cH}^{t_1,t_2,t_3}=({\cH}_1\mid {\cH}_2\mid {\cH}_3)$. Note that 
\begin{align*}
    {\cH}^{t_1,t_2,t_3+1}=\bigcup_{\lambda \in  \{0,1\}}(({\cH}_1, {\cH}_1\mid {\cH}_2, {\cH}_2\mid {\cH}_3, {\cH}_3)+\lambda(\zero, \one \mid \zero, \two \mid \zero, \four)).
\end{align*}
By Corollaries \ref{coro4} and  \ref{lemma2}, 
\begin{align}\label{eq:Ht3p1_}
    H^{t_1,t_2,t_3+1}&=\bigcup_{\lambda \in \{0,1\}}(\Phi({\cH}_1, {\cH}_1\mid {\cH}_2, {\cH}_2\mid {\cH}_3, {\cH}_3)+\lambda(\zero, \one,  \zero, \one,  \zero, \one))\nonumber\\
    &=A_0 \cup A_1,
\end{align}
where $A_\lambda=\Phi({\cH}_1, {\cH}_1\mid {\cH}_2, {\cH}_2\mid {\cH}_3, {\cH}_3)+\lambda(\zero, \one,  \zero, \one,  \zero, \one)$, ${\lambda \in \{0,1\}}$.
Next, we show that the minimum weight of $A_\lambda$ is $N'/2$.  Any element in $A_\lambda$ is of the form $\Phi(u_1, u_1\mid u_2, u_2 \mid u_3, u_3)+\lambda(\zero, \one,  \zero, \one,  \zero, \one)$, for $\uu=(u_1\mid u_2 \mid u_3) \in ({\cH}_1\mid {\cH}_2\mid {\cH}_3)$. Let $\uu=(u_1\mid u_2 \mid u_3) \in ({\cH}_1\mid {\cH}_2\mid {\cH}_3){\setminus}\{\zero\}$. 
When $\lambda =0$, we have that $\wt_H(\Phi(u_1, u_1\mid u_2, u_2 \mid u_3, u_3))=2\wt_H(\Phi(\uu))$. Thus, the minimum weight of $A_0$ 
is $2\cdot N/2=N'/2$. Otherwise, when $\lambda=1$, we have that $\wt_H(\Phi(u_1, u_1\mid u_2, u_2 \mid u_3, u_3)+(\zero, \one,  \zero, \one,  \zero, \one))= \wt_H(\Phi(\uu))+\alpha_1-\wt_H(u_1)+2\alpha_2-\wt_H(\Phi_2(u_2))+ 4\alpha_3-\wt_H(\Phi_3(u_3))=\wt_H(\Phi(\uu)) +\alpha_1+2\alpha_2+4\alpha_3-\wt_H(\Phi(\uu))=N=N'/2$. Thus, the minimum weight of $A_1$  
is $N'/2$. Therefore, from (\ref{eq:Ht3p1_}), the minimum weight of $H^{t_1,t_2,t_3+1}$ is $N'/2$.
\end{proof}

 \begin{theorem}\label{Th:H3}
Let $t_1\geq 1$, $t_2\geq 0$, and $t_3\geq 1$ be integers. The $\Z_2\Z_4 \Z_8$-additive code 
${\cH}^{t_1,t_2,t_3}$, generated by the matrix $A^{t_1,t_2,t_3}$, is a $\Z_2\Z_4 \Z_8$-additive Hadamard code. 
\end{theorem}
\begin{proof}
It follows from Propositions \ref{Th:H1} and \ref{Th:H2}.
\end{proof}

\begin{example}\label{ex:linear101}
The $\Z_2\Z_4 \Z_8$-additive code ${\cH}^{1,0,1}$ generated by the  matrix $A^{1,0,1}$, given in (\ref{eq:recGenMatrix0}),
is a $\Z_2\Z_{4}\Z_8$-additive Hadamard code of type $(2,1,1;1,0,1)$.  
We can write ${\cH}^{1,0,1}=\bigcup_{\alpha \in\Z_2} ({\cA} + \alpha \textbf{1})$, where ${\cA}=\{\lambda (0, 1\mid 1\mid 1): \lambda \in \Z_8\}$.
Thus, $H^{1,0,1}=\Phi({\cH}^{1,0,1})= \bigcup_{\alpha \in\Z_2} (\Phi({\cA}) + \alpha \textbf{1})$, where $\Phi({\cA})$ consists of all the rows of the Hadamard matrix 
\begin{equation*}
    H(2,4)=
    \left(\begin{array}{cccccccc}
0 &0 &0 &0 &0 &0 &0 &0  \\
0 &1 &0 &1 &0 &1 &0 &1   \\
0 &0 &1 &1 &0 &0 &1 &1  \\
0 &1 &1 &0 &0 &1 &1 &0 \\
0 &0 &0 &0 &1 &1 &1 &1 \\
0 &1 &0 &1 &1 &0 &1 &0 \\
0 &0 &1 &1 &1 &1 &0 &0 \\
0 &1 &1 &0 &1 &0 &0 &1 \\
\end{array}\right).
\end{equation*}
Note that $\Phi({\cA})$ is linear and the minimum distance of $\Phi({\cA})$ is $4$, so $H^{1,0,1}$ is a binary linear Hadamard code of length $8$. 
\end{example}


\begin{proposition}\label{HadamardEq}
 Let $t_1\geq 1$, $t_2\geq 0$, and $t_3\geq 1$ be integers. Let $H^{t_1,t_2,t_3}$ be a $\Z_2\Z_4 \Z_8$-linear Hadamard code of length $2^t$. Then, $t+1=3t_1+2t_2+t_3$. 
\end{proposition}
\begin{proof}
Since $H^{t_1,t_2,t_3}$ is a binary Hadamard code of length $2^t$, we have that $|H^{t_1,t_2,t_3}|=2\cdot 2^t=2^{t+1}$. Note that $|H^{t_1,t_2,t_3}|=2^{3t_1+2t_2+t_3}$, and hence $t+1=3t_1+2t_2+t_3$. 
\end{proof}

\medskip
 Now, we recall the following theorem in order to compare the $\Z_2\Z_4\Z_8$-linear Hadamard codes (with $\alpha_1\neq 0$, $\alpha_2\neq 0$ and $\alpha_3\neq 0$) with the $\Z_2\Z_4$-linear Hadamard codes (with $\alpha_1\neq 0$, $\alpha_2\neq 0$).

\begin{theorem}\cite{PRV06}\label{rankZ2Z4}
Let $t\geq 3$ and $t_2 \in \{0,\dots,\lfloor t/2 \rfloor\}$. Let $H^{t_1,t_2}$ be the nonlinear $\Z_2\Z_4$-linear Hadamard code of length $2^t$ and type $(\alpha_1, \alpha_2;t_2,t_3)$, where $\alpha_1=2^{t-t_1}$, $\alpha_2=2^{t-1}-2^{t-t_1-1}$, and $t_2=t+1-2t_1$. Then, $$\rank(H^{t_1,t_2})=t_2+2t_1+ \binom{t_1}{2} \ \mbox{ and } \ \kernel(H^{t_1,t_2})=t_1+t_2.$$
\end{theorem}

Also, we recall the construction of the $\Z_{2^s}$-linear Hadamard codes with $s\geq 2$ studied in \cite{KernelZ2s}, and the following theorem given in \cite{EquivZ2s}, in order to compare these codes with the $\Z_2\Z_4\Z_8$-linear Hadamard codes having $\alpha_1\neq 0$, $\alpha_2\neq 0$, and $\alpha_3\neq 0$. Let $T_i=\lbrace j\cdot 2^{i-1}\, :\, j\in\lbrace0,1,\dots,2^{s-i+1}-1\rbrace \rbrace$ for all $i \in \{1,\ldots,s \}$.
Note that $T_1=\lbrace0,\dots,2^{s}-1\rbrace$. Let $t_1$, $t_2$,\dots,$t_s$ be non-negative integers with $t_1\geq1$. Consider the matrix $\bar{A}^{t_1,\dots,t_s}$ whose columns are exactly all the vectors of the form $\mathbf{z}^T$, $\mathbf{z}\in\lbrace1\rbrace\times T_1^{t_1-1}\times T_{2}^{t_2}\times\cdots\times T_s^{t_s}$. Let $\bar{\mathcal{H}}^{t_1,\dots,t_s}$ be the $\Z_{2^s}$-additive code of type $(n;t_1,\dots,t_s)$  generated by the matrix $\bar{A}^{t_1,\dots,t_s}$. 
 Let $\bar{H}^{t_1,\dots,t_s}=\Phi(\bar{\mathcal{H}}^{t_1,\dots,t_s})$ be the corresponding $\Z_{2^s}$-linear Hadamard code.

\begin{theorem}\cite{EquivZ2s}\label{theo:equi}
Let $\bar{H}^{t_1,\dots,t_s}$ be the $\Z_{2^s}$-linear Hadamard code, with $s\geq 2$ and $t_s\geq 1$.
Then,  for all $\ell\in\{1,\dots,t_s\}$, $\bar{H}^{t_1,\dots,t_s}$  is permutation equivalent to the
$\Z_{2^{s+\ell}}$-linear Hadamard code
$\bar{H}^{1,\zero^{\ell-1},t_1-1,t_2,\dots,t_{s-1},t_s-\ell}.$
\end{theorem}

For $5\leq t\leq 11$, Tables 1 and 3 given in \cite{KernelZ2s}  show all possible values of ($t_1,\dots,t_s$) corresponding to nonlinear $\Z_{2^s}$-linear Hadamard codes, with $s\geq 2$, of length $2^t$. For each of them, the values $(r,k)$ are shown, where $r$ is the rank  and $k$ is the dimension of the kernel.  Note that if two codes have different values $(r,k)$, they are not equivalent. 
 The following example shows that all the nonlinear $\Z_2\Z_4\Z_8$-linear Hadamard codes of length $2^{11}$, with $\alpha_1 \neq 0$, $\alpha_2 \neq 0$, and $\alpha_3 \neq 0$,  are not equivalent to any  $\Z_2\Z_4\Z_8$-linear Hadamard codes of any other type, any $\Z_2\Z_4$-linear Hadamard code, with $\alpha_1 \neq 0$ and $\alpha_2 \neq 0$, and any $\Z_{2^s}$-linear Hadamard code, with $s\geq 2$, of the same length $2^{11}$.

\begin{example}\label{ex:classi}
   Consider $t=11$. By solving equation $t+1=3t_1+2t_2+t_3$ given in Proposition \ref{HadamardEq}, all $\Z_2\Z_4\Z_8$-linear Hadamard codes  of length $2^{11}$ are the ones  in $$T=\{H^{1,0,9}, H^{1,1,7}, H^{1,2,5}, H^{1,3,3}, H^{1,4,1}, H^{2,0,6}, H^{2,1,4}, H^{2,2,2}, H^{3,0,3}, H^{3,1,1}\}.$$  By using Magma, their corresponding values of $(r,k)$,  where $r$ is the rank  and $k$ is the dimension of the kernel, are  $(12,12)$, $(14,9)$, $(17,8)$, $(21,7)$, $(26,6)$, $(17,8)$, $(22,7)$, $(28,6)$, $(28,6)$, and $(37,5)$, respectively. The code $H^{1,0,9}$ is the only linear code in $T$ since it has the same rank and the dimension of the kernel. By using Magma, we can check that the following codes in each pair are nonequivalent to each other: $(H^{1,2,5}, H^{2,0,6})$, $(H^{2,2,2}, H^{3,0,3})$. Therefore, none of the $\Z_2\Z_4\Z_8$-linear Hadamard codes of length $2^{11}$ is equivalent to another $\Z_2\Z_4\Z_8$-linear Hadamard code of any other type. 
   
   Let $\Bar{T}=T \setminus \{H^{1,0,9}\}$.  Similarly, by  solving equation $t+1=2t_1+t_2$ given in Theorem \ref{rankZ2Z4}, all nonlinear $\Z_2\Z_4$-linear Hadamard codes  of length $2^{11}$ are $H^{2,8}$, $H^{3,6}$, $H^{4,4}$ and $H^{5,2}$, and by Theorem \ref{rankZ2Z4}, their corresponding values of $(r,k)$ are $(13,10)$, $(15,9)$, $(18,8)$, and $(22,7)$, respectively. Note that if two codes have different values $(r,k)$, they are not equivalent.  By using Magma, we can check that $H^{2,1,4}$ and $H^{5,2}$ are nonequivalent. Therefore, all the codes in $\Bar{T}$ are nonequivalent to any $\Z_2\Z_4$-linear Hadamard code of length $2^{11}$. 
   
   Finally, note that all the codes in $\Bar{T}$, except $H^{1,1,7}$ and $H^{2,1,4}$,  are not equivalent to any $\Z_{2^s}$-linear Hadamard code, with $s\geq 2$, of length $2^{11}$, since they have different values of $(r,k)$. The $\Z_{2^s}$-linear Hadamard codes of length $2^{11}$, having the same values $(r,k)=(14,9)$ as $H^{1,1,7}$, are $\Bar{H}^{2,0,6}$, $\Bar{H}^{1,1,0,5}$, $\Bar{H}^{1,0,1,0,4}$, $\Bar{H}^{1,0,0,0,1,0,2}$, and $\Bar{H}^{1,0,0,0,0,0,1,0,0}$, which are equivalent to each other by Theorem \ref{theo:equi}. The $\Z_4$-linear Hadamard code $\Bar{H}^{6,0}$ is the only $\Z_{2^s}$-linear Hadamard code of length $2^{11}$, having the same values $(r,k)=(22,7)$ as $H^{2,1,4}$. However, by using Magma, we can check that the following codes in each pair are nonequivalent to each other: $(H^{1,1,7},\Bar{H}^{2,0,6})$, $(H^{2,1,4}, \Bar{H}^{6,0})$. 
   
   Therefore, all nonlinear $\Z_2\Z_4\Z_8$-linear Hadamard codes  of length $2^{11}$  are  not equivalent to any $\Z_2\Z_4\Z_8$-linear Hadamard code of any other type, any $\Z_2\Z_4$-linear and $\Z_{2^s}$-linear Hadamard code, with $s\geq 2$, of length $2^{11}$.
\end{example}

 \bibliography{Manuscript}
\end{document}